\documentclass[a4paper,UKenglish,cleveref, autoref, thm-restate]{lipics-v2019}
\title{Bounding the Mim-Width of Hereditary Graph Classes}

\titlerunning{Bounding the Mim-Width of Hereditary Graph Classes}

\author{Nick Brettell}{School of Mathematics and Statistics, Victoria University of Wellington, New Zealand}{nick.brettell@vuw.ac.nz}{https://orcid.org/0000-0002-1136-418X}{}

\author{Jake Horsfield}{School of Computing, University of Leeds, Leeds, UK}{sc15jh@leeds.ac.uk}{https://orcid.org/0000-0002-4388-5123}{}

\author{Andrea Munaro}{School of Mathematics and Physics, Queen's University Belfast, UK}{a.munaro@qub.ac.uk}{https://orcid.org/0000-0003-1509-8832}{}

\author{Giacomo Paesani}{Department of Computer Science, Durham University, UK}{giacomo.paesani@durham.ac.uk}{https://orcid.org/0000-0002-2383-1339}{}

\author{Dani\"el Paulusma}{Department of Computer Science, Durham University, UK}{daniel.paulusma@durham.ac.uk}{https://orcid.org/0000-0001-5945-9287}{}

\authorrunning{N. Brettell, J. Horsfield, A. Munaro, G. Paesani, and D. Paulusma}

\Copyright{Nick Brettell, Jake Horsfield, Andrea Munaro, Giacomo Paesani, and Dani\"el Paulusma}

\ccsdesc[100]{Theory of computation → Graph algorithms analysis}

\keywords{Width parameter, mim-width, clique-width, hereditary graph class}

\funding{The research in this paper received support from the Leverhulme Trust (RPG-2016-258).}

\acknowledgements{}

\nolinenumbers 

\hideLIPIcs

\usepackage{amsmath,amssymb,amsthm,xspace,xcolor,tikz,graphicx}
\usetikzlibrary{fit,automata,arrows,positioning}

\newtheorem{open}{Open Problem}

\newcommand{\NP}{{\sf NP}}
\newcommand{\XP}{{\sf XP}}

\newcommand{\W}{{\sf W}}
\newcommand{\cutmim}{\mathrm{cutmim}}
\newcommand{\mimw}{\mathrm{mimw}}

\newcommand{\tw}{\mathrm{tw}}
\newcommand{\net}{{\sf net}}
\newcommand{\bull}{{\sf bull}}
\renewcommand{\diamond}{{\sf diamond}}
\newcommand{\gem}{{\sf gem}}
\newcommand{\paw}{{\sf paw}}
\newcommand{\hammer}{{\sf hammer}}
\renewcommand{\bowtie}{{\sf bowtie}}
\newcommand{\sun}{{\sf sun}}
\newcommand{\ssi}{\subseteq_i}
\newcommand{\si}{\supseteq_i}

\begin{document}

\maketitle

\begin{abstract} 
A large number of \NP-hard graph problems are solvable in \XP\ time when parameterized by some width parameter.  Hence, when solving problems on special graph classes, it is helpful to know if the graph class under consideration has bounded width.  In this paper we consider mim-width,  a particularly general width parameter that has a number of algorithmic applications whenever a decomposition is ``quickly computable'' for the graph class under consideration.

We start by extending the toolkit for proving (un)boundedness of mim-width of graph classes.  By combining our new techniques with known ones  we then initiate a systematic study into bounding mim-width from the perspective of hereditary graph classes, and make a comparison with clique-width, a more restrictive width parameter that has been well studied.

We prove that for a given graph $H$, the class of $H$-free graphs has bounded mim-width if and only if it has bounded clique-width. We show that the same is not true for $(H_1,H_2)$-free graphs.

We identify several general classes of $(H_1,H_2)$-free graphs having unbounded clique-width, but bounded mim-width; moreover, we show that a branch decomposition of constant mim-width can be found in polynomial time for these classes.  Hence, these results have algorithmic implications: when the input is restricted to such a class of $(H_1,H_2)$-free graphs, many problems become polynomial-time solvable, including classical problems such as $k$-{\sc Colouring} and {\sc Independent Set}, domination-type problems known as LC-VSVP problems, and distance versions of LC-VSVP problems, to name just a few.
We also prove a number of new results showing that, for certain $H_1$ and $H_2$, the class of $(H_1,H_2)$-free graphs has unbounded mim-width.

Boundedness of clique-width implies boundedness of mim-width. By combining our results with the known bounded cases for clique-width, we present summary theorems of the current state of the art for the boundedness of mim-width for $(H_1,H_2)$-free graphs. In particular, we classify the mim-width of $(H_1,H_2)$-free graphs for all pairs $(H_1,H_2)$ with $|V(H_1)| + |V(H_2)| \le 8$.  When $H_1$ and $H_2$ are connected graphs, we classify all pairs $(H_1,H_2)$ except for one remaining infinite family and a few isolated cases.
\end{abstract}

\section{Introduction}
\label{sec:intro}

Many computationally hard graph problems can be solved efficiently after placing appropriate restrictions on the input graph. Instead of trying to solve individual problems in an ad hoc way, one may aim to find the underlying reasons why some sets of problems behave better on certain graph classes than other sets of problems. The ultimate goal in this type of research is to obtain complexity dichotomies for large families of graph problems. Such dichotomies tell us for which graph classes a certain problem or set of problems can or cannot be solved efficiently (under standard complexity assumptions).

One reason that might explain the jump from computational hardness to tractability after restricting the input to some graph class~${\cal G}$ is that ${\cal G}$ has bounded ``width'', that is, every graph in ${\cal G}$ has width at most~$c$ for some constant $c$.
One can define the notion of ``width'' in many different ways (see the surveys~\cite{GHOS08,Gu17,KLM09,Va12}).
As such, the various width parameters differ in strength. 
To explain this, we say that a width parameter~$p$ {\it dominates} a width parameter~$q$ if there is 
a function~$f$ such that $p(G)\leq f(q(G))$ for all graphs~$G$.
If~$p$ dominates~$q$ but $q$ does not dominate $p$, then~$p$ is said to be {\it more powerful} than~$q$.
As a consequence, proving that a problem is polynomial-time solvable for graph classes for which $p$ is bounded yields more tractable graph
classes than doing this for graph classes for which $q$ is bounded.
If both $p$ and $q$ dominate each other, then~$p$ and~$q$ are {\it equivalent}.
For instance, the width parameters boolean-width, clique-width, module-width, NLC-width and rank-width are all equivalent~\cite{BTV11,Jo98,OS06,Ra08}, but
more powerful than the equivalent parameters branch-width and treewidth~\cite{CO00,RS91,Va12}.
In this paper we focus on an even more powerful width parameter called {\it mim-width} (maximum induced matching width).  Vatshelle~\cite{Va12} introduced mim-width, which we define in Section~\ref{s-mim}, and proved that mim-width is more powerful than boolean-width, and consequently, clique-width, module-width, NLC-width and rank-width. 

\subsection{Algorithmic Implications}\label{sec:algimplications}

One trade-off of a more powerful width parameter is the difficulty in obtaining a branch decomposition of bounded width.
In general, computing mim-width is \NP-hard; deciding if the mim-width is at most~$k$ is \W$[1]$-hard when parameterized by $k$; and there is no polynomial-time algorithm for approximating the mim-width of a graph to within a constant factor of the optimal, unless $\mathsf{NP} = \mathsf{ZPP}$~\cite{SV16}.
Hence, in contrast to algorithms for graphs of bounded treewidth or clique-width, algorithms for graphs of bounded mim-width require a branch decomposition of constant mim-width as part of the input.
On the other hand, there are many interesting graph classes for which mim-width is bounded and {\it quickly computable}, that is, the class admits a polynomial-time algorithm for obtaining a branch decomposition of constant mim-width. We give examples of such graph classes known in the literature in Section~\ref{s-mimspecial} before discussing the new graph classes we found in Section~\ref{s-new}. Below we briefly discuss known algorithms for problems on graph classes of bounded mim-width.

Belmonte and Vatshelle~\cite{BV13} and  Bui-Xuan, Telle and Vatshelle~\cite{BTV13} proved that a large set of problems, known as  Locally Checkable Vertex Subset and Vertex Partitioning (LC-VSVP) problems~\cite{PT97}, can be solved in polynomial time
for graph classes where mim-width is bounded and quickly computable.
Well-known examples of such problems include {\sc (Total) Dominating Set}, {\sc Independent Set} and $k$-{\sc Colouring} for every fixed positive
integer~$k$.\footnote{In contrast to clique-width~\cite{KR03}, {\sc Colouring} (where $k$ is part of the input) is \NP-complete for graphs of bounded mim-width, as it is \NP-complete for circular-arc graphs~\cite{GJMP80}, which have mim-width at most~$2$~\cite{BV13}.} 
Later, Fomin, Golovach and Raymond~\cite{FGR18} proved that the \XP\ algorithms for {\sc Independent Set} and {\sc Dominating Set} are in a sense best possible, showing that these two problems are \W$[1]$-hard when parameterized by mim-width.

On the positive side, \XP\ algorithms parameterized by mim-width are now also known for problems outside the LC-VSVP framework. In particular, Jaffke, Kwon, Str{\o}mme and Telle~\cite{JKST19} proved that the distance versions of LC-VSVP problems can be solved in polynomial time
for graph classes where mim-width is bounded and quickly computable.
Jaffke, Kwon and Telle~\cite{JKT,JKT20} proved similar results for {\sc Longest Induced Path}, {\sc Induced Disjoint Paths}, $H$-{\sc Induced Topological Minor} and {\sc Feedback Vertex Set}.
The latter result has recently been generalized to {\sc Subset Feedback Vertex Set} and {\sc Node Multiway Cut}, by Bergougnoux, Papadopoulos and Telle~\cite{BPT20}.

Bergougnoux and Kant\'e~\cite{BK19} gave a meta-algorithm for problems with a global constraint, providing unifying \XP\ algorithms in mim-width for several of the aforementioned problems, as well as {\sc Connected Dominating Set}, {\sc Node Weighted Steiner Tree}, and {\sc Maximum Induced Tree}.
Galby, Munaro and Ries~\cite{GMR20} proved that {\sc Semitotal Dominating Set} is  polynomial-time solvable for graph classes where mim-width is bounded and quickly computable.

\subsection{Mim-width of Special Graph Classes}\label{s-mimspecial}

Belmonte and Vatshelle~\cite{BV13} proved that the mim-width of the following graph classes is bounded and quickly computable:
permutation graphs,  convex graphs 
and their complements, 
interval graphs
and their complements, 
circular $k$-trapezoid graphs, circular permutation graphs, Dilworth-$k$ graphs, $k$-polygon graphs, circular-arc graphs and complements of $d$-degenerate graphs.

Some of the results of~\cite{BV13} have been extended.
Let $K_r \boxminus K_r$ be the graph obtained from $2K_r$ by adding a perfect matching, and let $K_r \boxminus rP_1$ be the graph obtained from $K_r\boxminus K_r$ by removing all the edges in one of the complete graphs (see Section~\ref{sec:basis} for undefined notation).
Kang et al.~\cite{KKST17} showed that for any integer $r \ge 2$,
there is a polynomial-time algorithm for computing a branch decomposition of mim-width at most $r-1$ when the input is restricted to
$(K_r \boxminus rP_1)$-free chordal graphs, which generalize interval graphs, or $(K_r \boxminus K_r)$-free co-comparability graphs, which generalize permutation graphs. Hence,  in particular, all these classes have bounded mim-width.

Kang et al.~\cite{KKST17} also proved that the classes of chordal graphs, circle graphs and co-comparability graphs have unbounded mim-width; for the latter two classes, this was shown independently by Mengel~\cite{Me18}.
Vatshelle~\cite{Va12} and Brault-Baron et al.~\cite{BCM15} showed the same for grids and chordal bipartite graphs, respectively, whereas Mengel~\cite{Me18} proved that strongly chordal split graphs have unbounded mim-width.

Brettell et al.~\cite{BHP} showed that the mim-width of $(K_r,sP_1+P_5)$-free graphs is bounded and quickly computable for every $r\geq 1$ and $s\geq 0$. In particular, this yielded an alternative proof for showing that {\sc List $k$-Colouring} is polynomially solvable for $(sP_1+P_5)$-free graphs for all $k\geq 1$ and $s\geq 0$~\cite{CGKP15}. Let $K_{1,s}^1$ be the graph obtained from the $(s+1)$-vertex star $K_{1,s}$ after subdividing each edge once; note that $sP_1+P_5$ is an induced subgraph of $K_{1,s+2}^1$. In~\cite{BMP20}, the result of~\cite{BHP} on the mim-width of  $(K_r,sP_1+P_5)$-free graphs was generalized to $(K_r,K_{1,s}^1,P_t)$-free graphs. As a consequence, for all $k\geq 3$, $s\geq 1$ and $t\geq 1$, {\sc List $k$-Colouring} is polynomial-time solvable even for 
$(K_{1,s}^1,P_t)$-free graphs; previously this was shown for $k=3$ by Chudnovsky et al.~\cite{CSZ20}.

Brettell et al.~\cite{BMP20a} considered the following generalisation of convex graphs. A bipartite graph $G=(A,B,E)$ is ${\cal H}$-convex, for some family of graphs ${\cal H}$, if there exists a graph $H\in {\cal H}$ with $V(H)=A$ such that the set of neighbours in $A$ of each $b\in B$ induces a connected subgraph of $H$ (when ${\mathcal H}$ is the set of paths, we obtain exactly convex graphs). They showed that the class of ${\mathcal H}$-convex graphs has bounded and quickly computable mim-width if ${\mathcal H}$ is the set of cycles, or ${\mathcal H}$ is the set of trees with bounded maximum degree and bounded number of vertices of degree at least~$3$.

\subsection{Our Focus} 

We continue the study on boundedness of mim-width and aim to identify more graph classes of bounded or unbounded mim-width. Our motivation is both algorithmic and structural. As discussed above, there are clear algorithmic benefits if a graph class has bounded mim-width.
From a structural point of view, we aim to initiate a {\it systematic} study of the boundedness of mim-width, comparable to a similar, 
long-standing study of the boundedness of clique-width 
(see~\cite{DJP19,Gu17,KLM09} for some surveys on clique-width).

The framework of hereditary graph classes is highly suitable for such a study. 
A graph class ${\cal G}$ is {\it hereditary} if it is closed under vertex deletion. A class ${\cal G}$ is hereditary if and only if there exists a 
(unique) set of graphs ${\cal F}$ of 
(minimal) 
forbidden induced subgraphs for ${\cal G}$. That is, a graph $G$ belongs to ${\cal G}$ if and only if $G$ does not contain any graph from ${\cal F}$ as an induced subgraph. We also say that $G$ is 
{\it ${\cal F}$-free}. 
Note that ${\cal F}$ may have infinite size. For example, if ${\cal G}$ is the class of bipartite graphs, then ${\cal F}$ is the set of all odd cycles. 

As a natural starting point we consider the case where $|{\cal F}|=1$, say ${\cal F}=\{H\}$. It is not difficult to verify that a class of $H$-free graphs has bounded mim-width if and only if it has bounded clique-width if and only if $H$ is an induced subgraph of the $4$-vertex path $P_4$; see Section~\ref{sec:basis2} for details. On the other hand, there exist hereditary graph classes, such as interval graphs and permutation graphs, that have bounded mim-width, even mim-width $1$ \cite{Va12}, but unbounded clique-width \cite{GR00}. However, these graph classes have an {\it infinite} set of forbidden induced subgraphs. Hence, questions we aim to address in this paper are: Does there exist a hereditary graph class characterized by a finite set ${\cal F}$ that has bounded mim-width but unbounded clique-width? Can we use the same techniques as when dealing with clique-width? In particular we focus on the case 
where $|{\cal F}|=2$. Such classes are called {\it bigenic}.

\subsection{Our Results and Methodology}\label{s-new}

In order to work with width parameters it is useful to have a set of graph operations that {\it preserve} boundedness or unboundedness of the width parameter. That is, if we apply such a width-preserving operation, or only apply it a constant number of times, the width of the graph does not change by too much. In this way one might be able to modify an arbitrary graph from a given ``unknown'' class ${\cal G}_1$ into a graph from a class ${\cal G}_2$ known to have bounded or unbounded width. This would then imply that ${\cal G}_1$ also has bounded or unbounded width, respectively. 
Two useful operations preserving clique-width are vertex deletion~\cite{LR04} and subgraph complementation~\cite{KLM09}. The latter operation replaces every edge in some subgraph of the graph by a non-edge, and vice versa.
As we will see in \cref{s-soa}, subgraph complementation does not preserve boundedness or unboundedness of mim-width\footnote{The situation is different for mim-width 1; Vatshelle~\cite{Va12} showed that if $\mimw(G) = 1$ then $\mimw(\overline{G}) = 1$.}.

To work around this limitation, we collect and generalize known mim-width preserving graph operations from the literature in Section~\ref{s-mim} (some of these operations only show that the mim-width cannot decrease after applying them).
In the same section we also state some known useful results on mim-width and prove that elementary graph classes, such as walls and net-walls, have unbounded mim-width.   

In Sections~\ref{sec:bounded} and~\ref{sec:unbounded} we use the results from Section~\ref{s-mim}. In Section~\ref{sec:bounded} we present new bigenic classes of bounded mim-width. These 
graph 
classes are all known to  have unbounded clique-width. Hence, our results show that 
the dichotomy for boundedness of mim-width no longer coincides with the one for clique-width when $|{\cal F}|=2$ instead of $|{\cal F}|=1$. Moreover, for each of these classes, a branch decomposition of constant mim-width is easily computable for any graph in the class.  This immediately implies that there are polynomial-time algorithms for many problems when restricted to these classes, as described in \cref{sec:algimplications}. 
In Section~\ref{sec:unbounded} we present new bigenic classes of unbounded mim-width; these graph classes are known to have unbounded clique-width.

In Section~\ref{s-soa} we give a state-of-the-art summary of our new results combined with known results. The known results include the bigenic graph classes of bounded clique-width (as bounded clique-width implies bounded mim-width). In the same section we compare our results for the mim-width of bigenic graph classes with the ones for clique-width. We also state a number of open problems.

\section{Preliminaries}
\label{sec:basis}

We consider only finite graphs $G=(V,E)$ with no loops and no multiple edges. 
For a vertex $v \in V$, the \textit{neighbourhood} $N(v)$ is the set of vertices adjacent to $v$ 
in $G$.
The \textit{degree} $d(v)$ of a vertex $v \in V$ is the size $|N(v)|$ of its neighbourhood.
A graph is \textit{subcubic} if every vertex has degree at most~$3$.
For disjoint $S,T \subseteq V$, we say that
 $S$ is \textit{complete to} $T$ if every vertex of $S$ is adjacent to every vertex of $T$, and $S$ is \textit{anticomplete to} $T$ if there are no edges between $S$ and $T$.
The \textit{distance} from a vertex $u$ to a vertex $v$ in $G$ is the length of a shortest path between $u$ and $v$. 
A set $S\subseteq V$ {\it induces} the subgraph 
$G[S]=(S,\{uv\; :\; u,v\in S, uv\in E\})$.
If $G'$ is an induced subgraph of $G$ we write $G'\ssi G$. 
The \textit{complement} of $G$ is the graph $\overline{G}$ with vertex set $V(G)$, such that $uv \in E(\overline{G})$ if and only if $uv \notin E(G)$.

Given a graph $G$ and a degree-$k$ vertex $v$ of $G$ with $N(v) = \{u_1,\dots,u_k\}$, the \textit{clique implant on $v$} is the operation of deleting $v$, adding $k$ new vertices $v_1, \dots, v_k$ forming a clique, and adding edges $v_iu_i$ for each $i \in \{1,\dots,k\}$. 
The \textit{$k$-subdivision} of an edge $uv$ in a graph replaces $uv$ by $k$ new vertices $w_{1}, \dots, w_{k}$ with edges $uw_{1}, w_{k}v$ and $w_{i}w_{i+1}$ for each $i \in \{1, \dots, k-1\}$, i.e. the edge is replaced by a path of length $k + 1$. 
The \textit{disjoint union} $G+H$ of graphs $G$ and $H$ has vertex set $V(G) \cup V(H)$ and edge set $E(G) \cup E(H)$.
We denote the disjoint union of $k$ copies of $G$ by $kG$. 
For a graph $H$, a graph $G$ is {\it $H$-free} if $G$ has no induced subgraph isomorphic to $H$. 
For a set of graphs $\{H_1,\ldots,H_k\}$, a graph $G$ is {\it $(H_1,\ldots,H_k)$-free} if $G$ is $H_i$-free for every $i\in \{1,\ldots,k\}$. 

An \textit{independent set} of a graph is a set of pairwise non-adjacent vertices. 
A \textit{clique} of a graph is a set of pairwise adjacent vertices.
A \textit{matching} of a graph is a set of pairwise non-adjacent edges.
A matching $M$ of a graph $G$ is \textit{induced} if there are no edges of $G$ between vertices incident to distinct edges of $M$.

The path, cycle and complete graph on $n$ vertices are denoted by $P_n$, $C_n$ and $K_{n}$, respectively.
The graph $K_3$ is also called the \textit{triangle}. A graph is \textit{$r$-partite}, for $r \geq 2$, if its vertex set admits a partition into $r$ classes such that every edge has its endpoints in different classes. An $r$-partite graph in which every two vertices from different partition classes are adjacent is a \textit{complete $r$-partite graph} and a $2$-partite graph is also called \textit{bipartite}. A graph is \textit{co-bipartite} if it is the complement of a bipartite graph.
A \textit{split graph} is a graph $G$ that admits a {\it split partition} $(C,I)$, that is, $V(G)$ can be partitioned into a clique $C$ and an independent set $I$. Equivalently, a graph is split if and only if it is $(2P_2,C_4,C_5)$-free.
The {\it subdivided claw}~$S_{h,i,j}$, for $1\leq h\leq i\leq j$ is the tree with one vertex~$x$ of degree~$3$ and exactly three leaves, which are of distance~$h$,~$i$ and~$j$ from~$x$, respectively.
Note that $S_{1,1,1}=K_{1,3}$. For $t\geq 3$, $\sun_t$ denotes the graph on $2t$ vertices obtained from a complete graph on $t$ vertices $u_1,\ldots,u_t$ by adding $t$ vertices $v_1,\ldots,v_t$ such that $v_i$ is adjacent to $u_i$ and $u_{i+1}$ for each $i \in \{1,\dotsc,t-1\}$ and $v_t$ is adjacent to $u_1$ and $u_t$. See Figure~\ref{f-sun} for a picture of $\sun_5$.

\begin{figure}[h]
\begin{center}
\begin{minipage}{0.33\textwidth}
\centering
\begin{tikzpicture}[scale=0.35]
\draw (0,2)--(1.9,0.6)--(1.16,-1.6)--(-1.16,-1.6)--(-1.9,0.6)--(0,2)--(-1.16,-1.6)--(1.9,0.6)--(-1.9,0.6)--(1.16,-1.6)--(0,2)
--(2.14,2.91)--(1.9,0.6)--(3.43,-1.14)--(1.16,-1.6)--(0,-3.6) (0,2)--(-2.14,2.91)--(-1.9,0.6)--(-3.43,-1.14)--(-1.16,-1.6)--(0,-3.6);
\draw[fill=white] (-1.9,0.6) circle [radius=5pt] (1.9,0.6) circle [radius=5pt] (0,2) circle [radius=5pt] (-1.16,-1.6) circle [radius=5pt] (1.16,-1.6) circle [radius=5pt] (3.43,-1.14) circle [radius=5pt] (-3.43,-1.14) circle [radius=5pt]
(0,-3.6) circle [radius=5pt] (-2.14,2.91) circle [radius=5pt] (2.14,2.91) circle [radius=5pt]; 
\end{tikzpicture} 
\end{minipage} 
\caption{The graph $\sun_5$.}\label{f-sun}
\end{center} 
\end{figure}
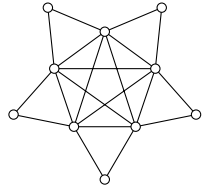

\section{Mim-Width: Definition and Basic Results}\label{s-mim}
\label{sec:basis2}

A \textit{branch decomposition} for a graph $G$ is a pair $(T, \delta)$, where $T$ is a subcubic tree and $\delta$ is a bijection from~$V(G)$ to the leaves of $T$. Each edge $e \in E(T)$ naturally partitions the leaves of $T$ into two classes, depending on which component they belong to when $e$ is removed. 
In this way, each edge $e \in E(T)$ corresponds to a partition $L_e$ and $\overline{L_e}$ of the set of leaves of $T$, depending on which component of $T-e$ the leaves of $T$ belong to.
Consequently, each edge $e$ induces a partition $(A_e, \overline{A_e})$ of $V(G)$, where $\delta(A_e) = L_e$ and $\delta(\overline{A_e}) = \overline{L_e}$.
For two disjoint sets $X$ and $Y$,
let $G[X, Y]$ denote the bipartite subgraph of $G$ induced by the edges with one endpoint in $X$ and the other in $Y$. 
For each edge $e \in E(T)$ and corresponding partition $(A_{e}, \overline{A_{e}})$ of $V(G)$, we denote by $\cutmim_{G}(A_{e}, \overline{A_{e}})$ the size of a maximum induced matching in $G[A_{e}, \overline{A_{e}}]$. The \emph{mim-width} of the branch decomposition $(T, \delta)$ is the quantity $\mimw_{G}(T, \delta) = \max_{e \in E(T)}\cutmim_{G}(A_{e}, \overline{A_{e}})$. The \emph{mim-width} of the graph $G$, denoted $\mimw(G)$, is the minimum value of $\mimw_{G}(T, \delta)$ over all possible branch decompositions $(T, \delta)$ for $G$.
See Figure~\ref{l-exx} for an example.

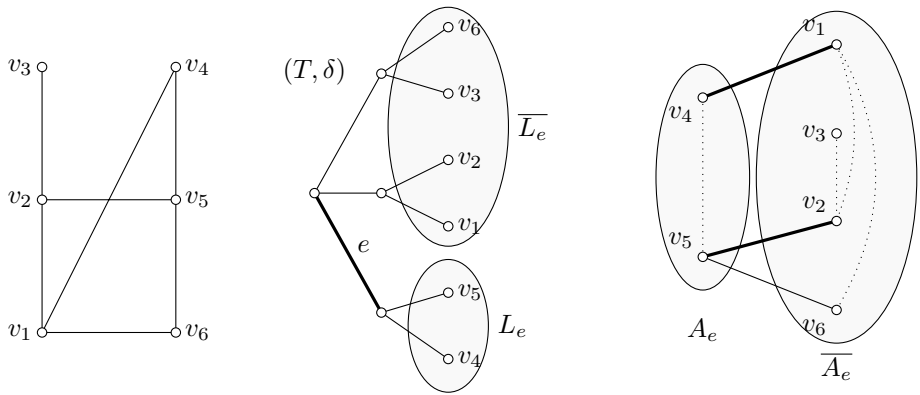
\begin{figure}[h]
\begin{minipage}[c]{0.2\textwidth}
\begin{tikzpicture}[scale=0.88] \draw(-1,-2)--(1,2)--(1,-2)--(-1,-2)--(-1,2)(1,0)--(-1,0);
\draw[fill=white] (1,2) circle [radius=2pt] (1,0) circle [radius=2pt] (1,-2) circle [radius=2pt]
(-1,2) circle [radius=2pt] (-1,0) circle [radius=2pt] (-1,-2) circle [radius=2pt];
\node[right] at (1,-2) {$v_6$}; \node[right] at (1,0) {$v_5$}; \node[right] at (1,2) {$v_4$};
\node[left] at (-1,2) {$v_3$}; \node[left] at (-1,0) {$v_2$}; \node[left] at (-1,-2) {$v_1$};
\end{tikzpicture}
\end{minipage}
\qquad
\begin{minipage}[c]{0.3\textwidth}
\begin{tikzpicture}[scale=0.88] \draw[color=black, fill=gray!5] (1,-2) ellipse (0.6cm and 1cm);
\draw[color=black, fill=gray!5] (1,1) ellipse (0.9cm and 1.8cm); \draw[very thick] (-1,0)--(0,-1.8);
\draw(1,2.5)--(0,1.8)--(-1,0)--(0,0)--(1,0.5)(1,-2.5)--(0,-1.8)--(1,-1.5) (1,-0.5)--(0,0) (1,1.5)--(0,1.8);
\draw[fill=white] (1,2.5) circle [radius=2pt] (1,1.5) circle [radius=2pt] (1,0.5) circle [radius=2pt]
(1,-2.5) circle [radius=2pt] (1,-1.5) circle [radius=2pt] (1,-0.5) circle [radius=2pt]
(0,1.8) circle [radius=2pt](0,-1.8) circle [radius=2pt](0,0) circle [radius=2pt] (-1,0) circle [radius=2pt];
\node[right] at (1,2.5) {$v_6$}; \node[right] at (1,1.5) {$v_3$}; \node[right] at (1,0.5) {$v_2$};
\node[right] at (1,-0.5) {$v_1$}; \node[right] at (1,-1.5) {$v_5$}; \node[right] at (1,-2.5) {$v_4$};
\node[right] at (-0.5,-0.8) {$e$}; \node[right] at (1.6,-2) {$L_e$};
\node[right] at (1.9,1) {$\overline{L_e}$}; \node[above] at (-1,1.5) {$(T,\delta)$};
\end{tikzpicture}
\end{minipage}
\qquad
\begin{minipage}[c]{0.07\textwidth}
\begin{tikzpicture}[scale=0.88] \draw[color=black, fill=gray!5] (-1,0) ellipse (0.7cm and 1.7cm);
\draw[color=black, fill=gray!5] (1,0) ellipse (1.2cm and 2.5cm);
\draw[very thick] (1,2)--(-1,1.2)(1,-0.66)--(-1,-1.2); \draw (1,-2)--(-1,-1.2);
\draw[dotted] (-1,1.2)--(-1,-1.2) (1,0.66)--(1,-0.66)to[out=70,in=290](1,2)to[out=300,in=60](1,-2);
\draw[fill=white] (1,2) circle [radius=2pt] (1,0.66) circle [radius=2pt] (1,-0.66) circle [radius=2pt]
(1,-2) circle [radius=2pt] (-1,1.2) circle [radius=2pt] (-1,-1.2) circle [radius=2pt];
\node[below left] at (-1,1.2) {$v_4$}; \node[above left] at (-1,-1.2) {$v_5$}; \node[above left] at (1,2) {$v_1$};
\node[left] at (1,0.66) {$v_3$}; \node[above left] at (1,-0.66) {$v_2$}; \node[below left] at (1,-2) {$v_6$};
\node[below] at (-1,-2) {$A_e$}; \node[below] at (1,-2.5) {$\overline{A_e}$};
\end{tikzpicture}
\end{minipage}
\caption{An example of a graph $G$ with a branch decomposition $(T,\delta)$. It can be easily seen that $\mimw_{G}(T, \delta)\leq 2$. The partition $(A_{e}, \overline{A_{e}})$ of $V(G)$ in the rightmost figure witnesses that $\mimw_{G}(T, \delta)\geq 2$. Hence, $\mimw_{G}(T, \delta)=2$. It can be checked that the branch decomposition $(T',\delta')$ obtained from $(T,\delta)$ by swapping 
$v_2$ and $v_5$ and swapping $v_3$ and $v_4$ shows that $\mimw(G)=1$.}\label{l-exx}
\end{figure}

\subsection{Mim-Width Preserving Operations}
The following three lemmas, the first of which is due to Vatshelle, show that vertex deletion, edge subdivision and clique implantation do not change the mim-width of a graph by too much. 

\begin{lemma}[\cite{Va12}] 
  \label{vertexdeletion}
Let $G$ be a graph and $v \in V(G)$. Then $\mimw(G) - 1 \leq \mimw(G-v) \leq \mimw(G)$.
\end{lemma}

\begin{lemma}\label{subdivision}
Let $G$ be a graph and let $G'$ be the graph obtained by $1$-subdividing an edge of $G$. Then $\mimw(G)\leq \mimw(G') \leq \mimw(G) + 1$. 
\end{lemma}

\begin{proof} 
Let $uv$ be the subdivided edge of $G$, and let $w \in V(G')\setminus V(G)$ such that $\{uw, wv\} \subseteq E(G')$. We first prove that $\mimw(G)\leq \mimw(G')$. Given a branch decomposition $(T',\delta')$ for $G'$, we construct a branch decomposition $(T,\delta)$ for $G$ such that $\mimw_{G}(T, \delta) \leq \mimw_{G'}(T', \delta')$.  Since $V(G') = V(G) \cup \{w\}$, we simply let $T$ be the tree obtained from $T'$ by deleting the leaf $\delta'(w)$, and let $\delta$ be the restriction of $\delta'$ to $V(G)$. Clearly, $(T,\delta)$ is a branch decomposition for $G$. 

We claim that $\mimw_{G}(T, \delta) \leq \mimw_{G'}(T', \delta')$. Suppose, to the contrary, that there exists $e \in E(T)$ such that $\cutmim_{G}(A_{e}, \overline{A_{e}}) > \mimw_{G'}(T', \delta')$, and let $M$ be a maximum induced matching in $G[A_{e}, \overline{A_{e}}]$. By construction, $e$ is also an edge of $T'$ and the partition $(B_{e}, \overline{B_{e}})$ of $V(G')$ corresponding to $e$ is either $(A_{e} \cup \{w\}, \overline{A_{e}})$ or $(A_{e}, \overline{A_{e}} \cup \{w\})$. If $uv \notin M$, then $M$ is also an induced matching in $G'[B_{e}, \overline{B_{e}}]$. On the other hand, if $uv \in M$, then either $M \setminus \{uv\} \cup \{uw\}$ or $M \setminus \{uv\} \cup \{wv\}$ is an induced matching in $G'[B_{e}, \overline{B_{e}}]$. In all cases, we find an induced matching in $G'[B_{e}, \overline{B_{e}}]$ of size $|M| = \cutmim_{G}(A_{e}, \overline{A_{e}}) > \mimw_{G'}(T', \delta')$, a contradiction.

We now prove that $\mimw(G') \leq \mimw(G) + 1$. Given a branch decomposition $(T,\delta)$ for $G$, we construct a branch decomposition $(T',\delta')$ for $G'$ such that $\mimw_{G'}(T', \delta') \leq \mimw_{G}(T, \delta) + 1$. Let $T'$ be the subcubic tree obtained by attaching two pendant vertices $x_{1}$ and $x_{2}$ to the leaf $\delta(u)$ of $T$, and let $\delta'(x) = \delta(x)$, for each $x \in V(G)\setminus\{u\}$, and $\delta'(u) = x_{1}$ and $\delta'(w) = x_{2}$. Clearly, $(T',\delta')$ is a branch decomposition for $G'$. 

We claim that $\mimw_{G'}(T', \delta') \leq \mimw_{G}(T, \delta) + 1$. Suppose, to the contrary, that there exists $e \in E(T')$ such that $\cutmim_{G'}(A_{e}, \overline{A_{e}}) > \mimw_{G}(T, \delta) + 1$. Clearly, $e \in E(T)$, for otherwise $\cutmim_{G'}(A_{e}, \overline{A_{e}}) \leq 1$. As $e$ is an edge of $T$, $u$ and $w$ belong to the same partition class of $V(G')$ and the partition $(B_{e}, \overline{B_{e}})$ of $V(G)$ corresponding to $e$ is obtained from $(A_{e}, \overline{A_{e}})$ by removing $w$. Let $M'$ be a maximum induced matching in $G'[A_{e}, \overline{A_{e}}]$. If $w$ is matched in $M'$, then it must be $wv \in M'$ and we remove this edge. If both $u$ and $v$ are matched in $M'$, we remove the matching edge incident to $u$. In all the other cases, we keep the matching edges. In this way we obtain an induced matching in $G[B_{e}, \overline{B_{e}}]$ of size at least $|M'| - 1 = \cutmim_{G'}(A_{e}, \overline{A_{e}}) - 1 > \mimw_{G}(T, \delta)$, a contradiction.
\end{proof}

\begin{lemma}\label{implant} Let $G$ be a graph and let $G'$ be the graph obtained from $G$ by a clique implant on $v \in V(G)$. Then $\mimw(G) \leq \mimw(G') \leq \mimw(G) + d(v)$.    
\end{lemma}

\begin{proof} 
We first prove that $\mimw(G) \leq \mimw(G')$. Suppose that $v$ is a degree-$k$ vertex of $G$ with $N(v) = \{u_1,\dots,u_k\}$ and let $\{v_1, \dots, v_k\}$ be the clique implanted on $v$. Given a branch decomposition $(T',\delta')$ for $G'$, we construct a branch decomposition $(T,\delta)$ for $G$ such that $\mimw_{G}(T, \delta) \leq \mimw_{G'}(T', \delta')$. Since $V(G') = V(G) \setminus \{v\} \cup \{v_{1}, \dots, v_{k}\}$, we build a tree $T$ as follows. We delete the leaves $\delta'(v_{2}), \dots, \delta'(v_{k})$ from $T'$ and let $\delta(x) = \delta'(x)$ if $x \in V(G) \setminus \{v\}$ and $\delta(v) = \delta'(v_{1})$. Clearly, $(T,\delta)$ is a branch decomposition for $G$.

We claim that $\mimw_{G}(T, \delta) \leq \mimw_{G'}(T', \delta')$. Suppose, to the contrary, that there exists $e \in E(T)$ such that $\cutmim_{G}(A_{e}, \overline{A_{e}}) > \mimw_{G'}(T', \delta')$ and let $M$ be a maximum induced matching in $G[A_{e}, \overline{A_{e}}]$. Suppose, without loss of generality, that $v \in A_{e}$. By construction, $e$ is also an edge of $T'$ and the partition $(B_{e}, \overline{B_{e}})$ of $V(G')$ corresponding to $e$ is of the form $((A_{e} \setminus \{v\}) \cup \{v_{1}\} \cup X, \overline{A_{e}} \cup Y)$, where $X \subseteq \{v_{2}, \dots, v_{k}\}$ and $Y = \{v_{2}, \dots, v_{k}\} \setminus X$. If $v$ is not matched in $M$, then $M$ is also an induced matching in $G'[B_{e}, \overline{B_{e}}]$ of size $|M| = \cutmim_{G}(A_{e}, \overline{A_{e}}) > \mimw_{G'}(T', \delta')$, a contradiction. Therefore, suppose that $v$ is matched in $M$. We have that $vu_{i} \in M$, for some $i \in \{1, \dots, k\}$. If $i = 1$, then $M$ is an induced matching in $G'[B_{e}, \overline{B_{e}}]$. Otherwise, $i > 1$ and we proceed as follows. If $v_{i}$ belongs to the partition class of $v_{1}$, we replace $M$ with $M \setminus \{vu_{i}\} \cup \{v_{i}u_{i}\}$. If $v_{i}$ does not belong to the partition class of $v_{1}$, we replace $M$ with $M \setminus \{vu_{i}\} \cup \{v_{1}v_{i}\}$. It is easy to see that in all cases we find an induced matching in $G'[B_{e}, \overline{B_{e}}]$ of size $|M| > \mimw_{G'}(T', \delta')$, a contradiction.

We now prove that $\mimw(G') \leq \mimw(G) + d(v)$. Suppose that $v$ is a degree-$k$ vertex of $G$, and let $\{v_1, \dots, v_k\}$ be the clique implanted on $v$. Given a branch decomposition $(T,\delta)$ for $G$, we construct a branch decomposition $(T',\delta')$ for $G'$ such that $\mimw_{G'}(T', \delta') \leq \mimw_{G}(T, \delta) + k$. We $(k-1)$-subdivide the edge of $T$ incident to $\delta(v)$ with new vertices $x_{1}, \dots, x_{k - 1}$, attach a pendant vertex $y_{i}$ to each $x_{i}$, let $\delta'(v_{k}) = \delta(v)$ and $\delta'(v_{i}) = y_{i}$, for each $i \in \{1, \dots, k - 1\}$, and finally let $\delta'(u) = \delta(u)$ for each $u \in V(G')\setminus \{v_{1}, \dots, v_{k}\}$. Clearly, $(T',\delta')$ is a branch decomposition for $G'$. 

We claim that $\mimw_{G'}(T', \delta') \leq \mimw_{G}(T, \delta) + k$. Suppose, to the contrary, that there exists $e \in E(T')$ such that $\cutmim_{G'}(A_{e}, \overline{A_{e}}) > \mimw_{G}(T, \delta) + k$. We have that $e \in E(T)$, for otherwise $\cutmim_{G'}(A_{e}, \overline{A_{e}}) \leq k$. But since $e$ is an edge of $T$, the vertices $v_{1}, \dots, v_{k}$ all belong to the same partition class of $V(G')$, say $A_{e}$, and the partition $(B_{e}, \overline{B_{e}})$ of $V(G)$ corresponding to $e$ is obtained from $(A_{e}, \overline{A_{e}})$ by removing $\{v_{1}, \dots, v_{k}\}$ and adding $v$ to $A_{e}$. Let $M'$ be a maximum induced matching in $G'[A_{e}, \overline{A_{e}}]$. By possibly removing the at most $k$ edges in $M'$ incident to vertices in $\{v_{1}, \dots, v_{k}\}$, we obtain an induced matching in $G[B_{e}, \overline{B_{e}}]$ of size at least $|M'| - k = \cutmim_{G'}(A_{e}, \overline{A_{e}}) - k > \mimw_{G}(T, \delta)$, a contradiction.
\end{proof}

Mengel~\cite{Me18} showed that adding edges inside the partition classes of a bipartite graph does not decrease mim-width by much. This result can be generalized to $k$-partite graphs in the following way.

\begin{lemma}\label{kpartite} Let $G$ be a $k$-partite graph with partition classes $V_{1}, \dots, V_{k}$, and let $G'$ be a graph obtained from $G$ by adding edges where for each added edge, there exists some $i$ such that both endpoints are in $V_i$. Then $\mimw(G') \geq \frac{1}{k} \cdot \mimw(G)$. 
\end{lemma}

\begin{proof} 
Let $(T,\delta)$ be a branch decomposition for $G'$. Since $G$ and $G'$ have the same vertex set, $(T,\delta)$ is a branch decomposition for $G$ as well. It is enough to show that $\mimw_{G}(T, \delta) \leq k \cdot \mimw_{G'}(T, \delta)$. Therefore, let $e \in E(T)$ be such that $\mimw_{G}(T, \delta) = \cutmim_{G}(A_{e}, \overline{A_{e}})$, and let $M$ be a maximum induced matching in $G[A_{e}, \overline{A_{e}}]$. For each $i$, consider the set $M_{i} = \{uv \in M: u \in A_{e} \cap V_{i}\}$. These $k$ sets partition $M$. Let $M'$ be a partition class of size at least $|M|/k$. Clearly, $M'$ is an induced matching in $G'[A_{e}, \overline{A_{e}}]$ and so $k \cdot \mimw_{G'}(T, \delta) \geq k\cdot|M'| \geq |M| = \mimw_{G}(T, \delta)$.
\end{proof}

The next lemma shows that to bound the mim-width of a class of graphs, we may restrict our attention to $2$-connected graphs in the class.
We note that this property is not specific to mim-width: it has also been observed, in \cite{GHOS08}, for rank-width, and this argument also applies for any appropriate width parameter defined using branch decompositions.
A \emph{block} is a maximal connected subgraph with no cut-vertex.

\begin{lemma}\label{l-2con}
    Let $G$ be a graph.  Then $\mimw(G) = \max \{\mimw(H) : H\textrm{ is a block of }G\}$.
    Moreover, 
    given branch decompositions $(T_H,\delta_H)$ of each block $H$ of $G$, with $\mimw_H(T_H,\delta_H) \le k$, we can compute a branch decomposition of $G$ with mim-width at most $k$ in polynomial time.
\end{lemma}

\begin{proof}
  By \cref{vertexdeletion}, $\mimw(G) \ge \max \{\mimw(H) : H\textrm{ is a block of }G\}$.
  We describe how to compute a branch decomposition $(T,\delta)$ of $G$ such that $\mimw_G(T,\delta) \le \max \{\mimw_H(T_H,\delta_H) : H\textrm{ is a block of }G\}$, in polynomial time.
  It suffices to describe a polynomial-time procedure when $G$ consists of two blocks $H_1$ and $H_2$ joined at a vertex $v$ (we can repeat this procedure $O(n)$ times, thereby constructing a branch decomposition for $G$ block-by-block).
  To construct $T$, join $T_{H_1}$ and $T_{H_2}$ by identifying the leaf $t_1 \in T_{H_1}$ and the leaf $t_2 \in T_{H_2}$ such that $\delta_{H_1}(v) = t_1$ and $\delta_{H_2}(v) = t_2$, and then create a new leaf $t$ incident to the identified vertex.
  Let $\delta$ inherit the mappings from $\delta_{H_1}$ and $\delta_{H_2}$, and set $\delta(v) = t$.
  If $e \in E(T)$ is incident to $t$, then $\cutmim_G(A_e,\overline{A_e}) \le 1$, since one of $A_e$ and $\overline{A_e}$ has size one.
  For any other edge of $T$, either $A_e$ or $\overline{A_e}$ contains $V(H_1)$ or $V(H_2)$.
  The result follows.
\end{proof}

The following lemma is due to Galby and Munaro, who used it to prove that {\sc Dominating Set} admits a PTAS for a subclass of VPG graphs when the representation is given.

\begin{lemma}[\cite{GM20}]\label{makeclique} Let $G$ be a graph and let $S \subseteq V$. Let $G'=(V',E')$ denote the graph with $V'=V$ and $E' = E \cup \{uv : u,v \in S\}$. Then $\mimw(G') \leq \mimw(G) + 1$.    
\end{lemma}

The final structural lemma is used to prove that $(sP_1+P_5,K_t)$-free graphs have bounded mim-width for every $s\geq 0$ and $t \ge 1$. It shows how we can bound the mim-width of a graph in terms of the mim-width of the graphs induced by blocks of a partition of the vertex set and the mim-width between any two of the parts. We include it here as it might be useful for bounding the mim-width of other graph classes.

\begin{lemma}[\cite{BHP}]
\label{mimmultijoin}
Let $G$ be a graph and $(X_1,\dotsc,X_p)$ be a partition of $V(G)$ such that $\cutmim_G(X_i,X_j) \le c$ for all distinct $i,j \in \{1,\dotsc,p\}$, and $p \ge 2$.
Then \[\mimw(G) \le \max\left\{c\left\lfloor\left(\frac{p}{2}\right)^2\right\rfloor,\max_{i \in \{1,\dotsc,p\}}\{\mimw(G[X_i])\} + c(p-1)\right\}.\]
Moreover, if $(T_i,\delta_i)$ is a branch decomposition of $G[X_i]$ for each $i$, then we can construct, in $O(1)$ time, a branch decomposition $(T,\delta)$ of $G$ with \[\mimw_G(T,\delta) \le \max\left\{c\left\lfloor\left(\frac{p}{2}\right)^2\right\rfloor,\max_{i \in \{1,\dotsc,p\}}\{\mimw_G(T_i,\delta_i)\} + c(p-1)\right\}.\]
\end{lemma}

\subsection{Mim-width of Some Basic Classes}
Recall that Vatshelle~\cite{Va12} showed that the class of grids has unbounded mim-width.
We next prove that the same holds for the class of walls, which we define momentarily.
Thus, we obtain a class of graphs with maximum degree~$3$ having unbounded mim-width, and we will use this result in order to prove \cref{split}.
Note that it also gives us a dichotomy, as graphs with maximum degree~$2$ have bounded clique-width and hence bounded mim-width.

A \emph{wall of height $h$ and width $r$} (an \emph{$(h \times r)$-wall} for short) is the graph obtained from the grid of height~$h$ and width~$2r$ as follows.
Let $C_1, \dots, C_{2r}$ be the set of vertices in each of the $2r$ columns of the grid, in their natural left-to-right order.
For each column $C_{j}$, let $e_{1}^{j}, e_{2}^{j}, \dots, e_{h-1}^{j}$ be the edges between two vertices of $C_j$, in their natural top-to-bottom order. If $j$ is odd, we delete all edges $e_{i}^{j}$ with $i$ even. If $j$ is even, we delete all edges $e_{i}^{j}$ with $i$ odd. We then remove all vertices of the resulting graph whose degree is~$1$. This final graph is an \textit{elementary $(h \times r)$-wall} and any subdivision of the elementary $(h \times r)$-wall is an \textit{$(h \times r)$-wall}.  For an example, see \cref{figwall}.

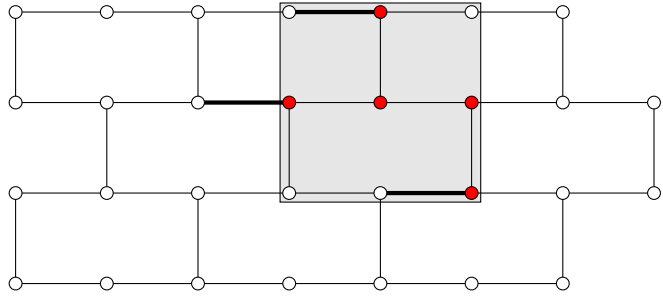
\begin{figure}
\begin{center}
\begin{tikzpicture}[scale=1.2]
\draw [fill=gray!20!white] (-0.1,-0.1) rectangle (2.1,2.1);
\draw (-3,2)--(3,2) (-3,1)--(4,1) (-3,0)--(4,0) (-3,-1)--(3,-1)
(-3,-1)--(-3,0) (-1,-1)--(-1,0) (1,-1)--(1,0) (3,-1)--(3,0)
(-2,0)--(-2,1) (0,0)--(0,1) (2,0)--(2,1) (4,0)--(4,1)
(-3,1)--(-3,2) (-1,1)--(-1,2) (1,1)--(1,2) (3,1)--(3,2);
\draw[ultra thick] (-1,1)--(0,1) (0,2)--(1,2)  (1,0)--(2,0);
\draw[fill=red]
(1,2) circle [radius=2pt] (0,1) circle [radius=2pt] (1,1) circle [radius=2pt] (2,1) circle [radius=2pt] (2,0) circle [radius=2pt];
\draw[fill=white] (-3,2) circle [radius=2pt] (-2,2) circle [radius=2pt] (-1,2) circle [radius=2pt] (2,2) circle [radius=2pt] (3,2) circle [radius=2pt] 
(0,2) circle [radius=2pt] (-3,1) circle [radius=2pt] (-2,1) circle [radius=2pt] (-1,1) circle [radius=2pt] (3,1) circle [radius=2pt] 
(4,1) circle [radius=2pt] (-3,0) circle [radius=2pt] (-2,0) circle [radius=2pt] (-1,0) circle [radius=2pt] (0,0) circle [radius=2pt] 
(1,0) circle [radius=2pt] (3,0) circle [radius=2pt] (4,0) circle [radius=2pt] (-3,-1) circle [radius=2pt](-2,-1) circle [radius=2pt] 
(-1,-1) circle [radius=2pt] (0,-1) circle [radius=2pt] (1,-1) circle [radius=2pt] (2,-1) circle [radius=2pt] (3,-1) circle [radius=2pt];
\end{tikzpicture}
\caption{An elementary $(4 \times 4)$-wall.  We illustrate an example of the case where $h \geq \sqrt[4]{n(W)/3}$ and $r < 2n$ in the proof of \cref{walls}:
$Q$ consists of the red vertices, $B$ is the the grey box, and the thick edges are a matching in $W[A_e,\overline{A_e}]$.}\label{figwall}
\end{center}
\end{figure}

\begin{theorem} \label{walls}
Let $W$ be an elementary $(n \times n)$-wall with $n \geq 7$. Then $\mimw(W) \geq \frac{\sqrt{n}}{50}$. In particular, the class of walls has unbounded mim-width.
\end{theorem}

\begin{proof} We let $n(W) = |V(W)| = 2n^{2} - 2$. Consider now a branch decomposition $(T,\delta)$ for $W$. There exists $e \in E(T)$ such that both partition classes $A_{e}$ and $\overline{A_{e}}$ of $V(W)$ contain at least $n(W)/3$ vertices
\cite[Lemma 2.3]{KKST17}. Kanj et al.~\cite[Lemma 4.10]{KPSX11} showed that if $G$ is a graph such that each of its subgraphs has average degree at most $d$, then any matching $M$ in $G$ contains an induced matching in $G$ of size at least $|M|/(2d-1)$. Since~$W$ is subcubic, it is sufficient to show that $W[A_{e}, \overline{A_{e}}]$ has a matching of size $\sqrt{n}/10$. We distinguish two cases, according to whether or not one of $W[A_{e}]$ and $W[\overline{A_{e}}]$ has a component of size at least $\sqrt{n(W)/3}$. 

Suppose first that $W[A_{e}]$ has a component $Q$ of size at least $\sqrt{n(W)/3}$. The component $Q$ is contained in a rectangle of the underlying $n \times 2n$ grid. Consider the smallest such rectangle $B$, i.e., the rectangle whose horizontal sides contain the uppermost and lowermost vertex in $Q$ and whose vertical sides contain the leftmost and rightmost vertex in $Q$. Let $h$ and $r$ be the height and width of $B$, respectively. Since $|V(Q)| \geq \sqrt{n(W)/3}$, one of $h$ and $r$ is at least $\sqrt[4]{n(W)/3}$. 

Suppose first that $h \geq \sqrt[4]{n(W)/3}$. If $r < 2n$, say without loss of generality $B$ does not intersect column~$C_{1}$, we do the following. For each row of $B$, consider the leftmost vertex of $Q$ in that row (since~$Q$ is connected, each row contains at least one vertex of $Q$). Clearly, the left neighbours of each such vertex belongs to $\overline{A_{e}}$, and so we have a matching in $W[A_{e}, \overline{A_{e}}]$ of size $h - 2 \geq \sqrt[4]{n(W)/3} - 2$, which is at least $\sqrt{n}/10$ when $n \geq 7$. If $r = 2n$, we distinguish two cases according to whether $h = n$ or not. In the first case (i.e., $r = 2n$ and $h = n$) we argue as follows. Since $Q$ is connected, each row of $B$ contains a vertex of $Q \subseteq A_{e}$. Moreover, there are at most $2n/3$ rows of $B$ with all vertices contained in $A_{e}$, for otherwise $|A_{e}| > (2n/3)\cdot 2n \geq 2n(W)/3$. So there are at least $n/3$ rows of $B$ containing a vertex of $A_{e}$ and a vertex of $\overline{A_{e}}$. We can therefore find a matching in $W[A_{e}, \overline{A_{e}}]$ of size at least $n/3$. In the second case (i.e., $r = 2n$ and $h < n$), we proceed as follows. We assume, without loss of generality, that $B$ does not intersect the uppermost row of the grid. We partition the columns of $B$ into disjoint layers containing two consecutive columns each. For each layer, we consider its left column and the uppermost vertex $v \in A_{e}$ therein (since $Q$ is connected, such a vertex exists). Let $v_{1}$ be the vertex on the grid above $v$, let $v_{2}$ be the vertex to the right of $v$ and let~$v_{3}$ be the vertex above $v_{2}$. By construction, $v_{1} \in \overline{A_{e}}$ and if $vv_{1} \in E(W)$, we select this edge. Otherwise, $vv_{1} \notin E(W)$ and so $v_{2}v_{3} \in E(W)$ and we have a path $vv_{2}v_{3}v_{1}$ in $W$ with $v \in A_{e}$ and $v_{1} \in \overline{A_{e}}$. We then select an edge of this path which belongs to $W[A_{e}, \overline{A_{e}}]$. Proceeding similarly for each layer, we obtain a matching in $W[A_{e}, \overline{A_{e}}]$ of size at least $r/2 = n$. Suppose finally that $h < \sqrt[4]{n(W)/3}$. We have that $r \geq \sqrt[4]{n(W)/3}$ and we proceed exactly as in the case $r = 2n$ and $h < n$ to obtain a matching in $W[A_{e}, \overline{A_{e}}]$ of size at least $r/2 \geq \sqrt[4]{n(W)/3}/2$. 

It remains to consider the situation in which all components of $W[A_{e}]$ and $W[\overline{A_{e}}]$ have size less than $\sqrt{n(W)/3}$. In particular, since $W[A_{e}]$ has more than $n(W)/3$ vertices, it has more than $\sqrt{n(W)/3}$ components. Let $Q_{1}, \dots, Q_{k}$ be these components. For each $i \in \{1, \dots, k\}$, there exists a vertex $u_{i} \in Q_{i}$ with a neighbour $v_{i} \in \overline{A_{e}}$, as $W$ is connected. Let $H$ be the subgraph of $W[A_{e}, \overline{A_{e}}]$ induced by $\{u_{1}, \dots, u_{k}\} \cup \{v_{1}, \dots, v_{k}\}$ (notice that we might have $v_{i} = v_{j}$ for some $i \neq j$). Let $H_{1}, \dots, H_{\ell}$ be the components of $H$ and let $n_{i} = |V(H_{i})|$, for each $i \in \{1, \dots, \ell\}$. By construction, $n_{i} \geq 2$, for each $i$. Moreover, since $H_{i}$ is a connected subcubic graph, it has a matching of size at least $(n_{i}-1)/3 \geq n_{i}/6$ \cite{BDDFK04}.
But then $H$ has a matching of size
\[\sum_{i = 1}^{\ell}\frac{n_{i}}{6} = \frac{|V(H)|}{6} \geq \frac{k}{6} \geq \frac{1}{6}\cdot\sqrt{\frac{n(W)}{3}}.\]
As in all cases we find a matching in $W[A_{e}, \overline{A_{e}}]$ of size at least $\frac{\sqrt{n}}{10}$, this concludes the proof.        
\end{proof}

\begin{corollary}
For an integer $\Delta$, let ${\cal G}_\Delta$ be the class of graphs of maximum degree at most~$\Delta$. 
Then the mim-width of ${\cal G}_\Delta$ is bounded if and only if $\Delta\leq 2$. 
\end{corollary}

A \emph{net-wall} is a graph that can be obtained from a wall~$G$ by performing a clique implant on each vertex of $G$ having degree three. An example of part of a net-wall is given in \cref{fig-netwall4col}. 

The following lemma is a straightforward consequence of \cref{walls,implant}.
\begin{lemma}
    \label{netwalls}
    The class of net-walls has unbounded mim-width.
\end{lemma}

Mengel~\cite{Me18} showed that strongly chordal split graphs, or equivalently $(\sun_3,\sun_4,\ldots)$-free split graphs, have unbounded mim-width.
We find two more subclasses of split graphs with unbounded mim-width by using \cref{subdivision,kpartite}.

\begin{lemma}
    \label{split}
    Let $\mathcal{G}$ be the class of split graphs, or equivalently $(C_4,C_5,2P_2)$-free graphs, where one of the following properties is satisfied by every $G \in \mathcal{G}$:
    \begin{enumerate}[(i)]
    \item $G$ has a split partition $(C,I)$ where each vertex in $I$ has degree~$2$ and each vertex in $C$ has at most three neighbours in $I$,
    \item $G$ has a split partition $(C,I)$ where each vertex in $I$ has degree at most $3$, and each vertex in $C$ has two neighbours in $I$, or
    \item $G$ is $\sun_t$-free $t\geq 3$.
   \end{enumerate} 
   Then $\mathcal{G}$ has unbounded mim-width.
\end{lemma}

\begin{proof}
Statement (iii) is due to Mengel~\cite{Me18}. To prove (i) and (ii), let $G$ be a wall, and let $G'$ be the graph obtained by $1$-subdividing each edge of $G$.  Partition $V(G')$ into $(A,B)$, where $B$ consists of the vertices of degree two introduced by the $1$-subdivisions.  Observe that $G'$ is bipartite, with vertex bipartition $(A,B)$. Let $G''$ be the graph obtained by making one of $A$ or $B$ a clique. By \cref{subdivision,kpartite}, $\mimw(G'') \ge \mimw(G)/2$. The result now follows from \cref{walls}.
\end{proof}

A graph is {\it chordal bipartite} if it is bipartite and every induced cycle has four vertices.
Brault-Baron et al.~\cite{BCM15} showed that the class of chordal bipartite graphs has unbounded mim-width
(we describe their construction in Section~\ref{sec:unbounded}).  
Combining their result with \cref{kpartite}, after adding all edges in a colour class, yields the following:
\begin{lemma}
    \label{cobipartite}
    The class of co-bipartite graphs, or equivalently $(3P_1,C_5,\overline{C_7},\overline{C_9},\ldots)$-free graphs, has unbounded mim-width.
\end{lemma}

As the last result in this section we consider hereditary classes defined by one forbidden induced subgraph.
It is folklore that the class of $H$-free graphs has bounded clique-width if and only if $H \subseteq_i P_4$ (see \cite{DP16} for a proof).
It turns out that the same dichotomy holds for mim-width.

\begin{theorem}\label{t-p4} The class of $H$-free graphs has bounded mim-width if and only if $H \subseteq_i P_4$.
\end{theorem}

\begin{proof} If $H \subseteq_i P_4$, then $H$-free graphs form a subclass of $P_4$-free graphs.
Every $P_4$-free graph has clique-width at most $2$~\cite{CO00} and so mim-width at most $2$~\cite{Va12}.
Suppose now that $H$ is a graph such that the class of $H$-free graphs has bounded mim-width. 
Recall that chordal bipartite graphs have unbounded mim-width~\cite{BCM15}
(see also Section~\ref{sec:unbounded}). Hence, $H$ is $C_3$-free. As 
co-bipartite graphs, and thus $3P_1$-free graphs, and split graphs, or equivalently, $(C_4,C_5,2P_2)$-free graphs, have unbounded mim-width by \cref{cobipartite,split}, this means that
$H$ is a $(3P_1,2P_2)$-free forest. It follows that $H\ssi P_4$. 
\end{proof}

\section{New Bounded Cases}
\label{sec:bounded}

In this section, we present three general classes and two further specific classes, of $(H_1,H_2)$-free graphs having bounded mim-width, but unbounded clique-width.
First, we present the three infinite families of classes of $(H_1,H_2)$-free graphs.
We show that for a class in one of these three families, there exists a constant $k$ such that for every graph $G$ in the class, and every $X \subseteq V(G)$, we have that $\cutmim_G(X, \overline{X}) \le k$.
This implies that every branch decomposition of $G$ has mim-width at most $k$.
Thus, {\it for a graph in one of these classes, a branch decomposition of constant mim-width is quickly computable}: any branch decomposition will suffice.
Finally, we present two more classes of $(H_1,H_2)$-free graphs having bounded mim-width, 
which do not have this property, but for which we prove that a branch decomposition of constant width can be computed in polynomial-time.

We make use of Ramsey theory.
By Ramsey's Theorem, for all positive integers $a$ and $b$, there exists an integer $R(a,b)$ such that if $G$ is a graph on at least $R(a,b)$ vertices, then $G$ has either a clique of size $a$, or an independent set of size $b$.

Recall that $K_r \boxminus K_r$ is the graph obtained from $2K_r$ by adding a perfect matching and that $K_r \boxminus rP_1$ is the graph obtained from $K_r\boxminus K_r$ by removing all the edges in one of the complete graphs.
We let $K_r \boxminus P_1$ denote the graph obtained from $K_r$ by adding a single vertex, attached to $K_r$ by a single pendant edge.
We also denote $\overline{C_4+P_1}$ as $\bowtie$.
Examples of these graphs are given in \cref{boundedegfigs}.

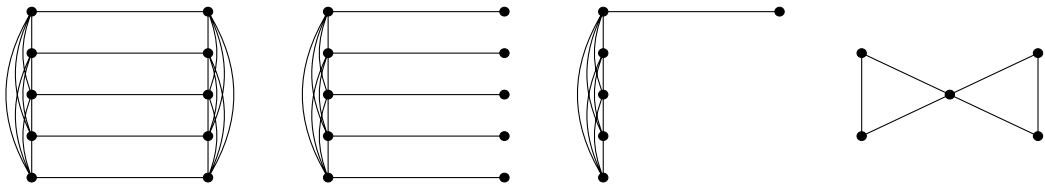
\begin{figure}[htb]
\begin{minipage}{0.22\textwidth}
\centering
\begin{tikzpicture}[xscale=0.58,yscale=.55]
\draw
(2,-2)--(-2,-2)--(-2,2)--(2,2) (-2,-1)--(2,-1) (-2,0)--(2,0) (-2,1)--(2,1) (2,-2)--(2,2)
(-2,-2) to[out=110,in=250] (-2,0) (-2,-1) to[out=110,in=250] (-2,1) (-2,0) to[out=110,in=250] (-2,2) (-2,-2) to[out=115,in=245] (-2,1)
(-2,-1) to[out=115,in=245] (-2,2) (-2,-2) to[out=120,in=240] (-2,2) (2,-2) to[out=70,in=290] (2,0) (2,-1) to[out=70,in=290] (2,1)
(2,0) to[out=70,in=290] (2,2) (2,-2) to[out=65,in=295] (2,1) (2,-1) to[out=65,in=295] (2,2) (2,-2) to[out=60,in=300] (2,2);
\draw[fill=black]
(-2,-2) circle [radius=3pt] (-2,-1) circle [radius=3pt] (-2,0) circle [radius=3pt] (-2,1) circle [radius=3pt] (-2,2) circle [radius=3pt] 
(2,-2) circle [radius=3pt] (2,-1) circle [radius=3pt] (2,0) circle [radius=3pt] (2,1) circle [radius=3pt] (2,2) circle [radius=3pt];
\end{tikzpicture}
\end{minipage}
\qquad
\begin{minipage}{0.20\textwidth}
\centering
\begin{tikzpicture}[xscale=0.58,yscale=.55]
\draw
(2,-2)--(-2,-2)--(-2,2)--(2,2) (-2,-1)--(2,-1) (-2,0)--(2,0) (-2,1)--(2,1)
(-2,-2) to[out=110,in=250] (-2,0) (-2,-1) to[out=110,in=250] (-2,1) (-2,0) to[out=110,in=250] (-2,2) 
(-2,-2) to[out=115,in=245] (-2,1) (-2,-1) to[out=115,in=245] (-2,2) (-2,-2) to[out=120,in=240] (-2,2);
\draw[fill=black]
(-2,-2) circle [radius=3pt] (-2,-1) circle [radius=3pt] (-2,0) circle [radius=3pt] (-2,1) circle [radius=3pt] 
(-2,2) circle [radius=3pt] (2,-2) circle [radius=3pt] (2,-1) circle [radius=3pt] (2,0) circle [radius=3pt] 
(2,1) circle [radius=3pt] (2,2) circle [radius=3pt];
\end{tikzpicture}
\end{minipage}
\qquad
\begin{minipage}{0.20\textwidth}
\centering
\begin{tikzpicture}[xscale=0.58,yscale=.55]
\draw
(-2,-2)--(-2,2)--(2,2)
(-2,-2) to[out=110,in=250] (-2,0) (-2,-1) to[out=110,in=250] (-2,1) (-2,0) to[out=110,in=250] (-2,2) 
(-2,-2) to[out=115,in=245] (-2,1) (-2,-1) to[out=115,in=245] (-2,2) (-2,-2) to[out=120,in=240] (-2,2);
\draw[fill=black]
(-2,-2) circle [radius=3pt] 
(-2,-1) circle [radius=3pt] 
(-2,0) circle [radius=3pt] 
(-2,1) circle [radius=3pt] 
(-2,2) circle [radius=3pt] 
(2,2) circle [radius=3pt];
\end{tikzpicture}
\end{minipage}
\qquad
\begin{minipage}{0.20\textwidth}
\centering
\begin{tikzpicture}[xscale=0.58,yscale=.55]
\draw (-2,-1)--(-2,1)--(0,0)--(-2,-1);
\draw (2,-1)--(2,1)--(0,0)--(2,-1);
\draw[fill=black]
(-2,-1) circle [radius=3pt] 
(-2,1) circle [radius=3pt] 
(0,0) circle [radius=3pt] 
(2,-1) circle [radius=3pt] 
(2,1) circle [radius=3pt];
\end{tikzpicture}
\end{minipage}
\caption{The graphs $K_5\boxminus K_5$, $K_5\boxminus 5P_1$, $K_5\boxminus P_1$, and $\bowtie = \overline{C_4+P_1}$.}
\label{boundedegfigs}
\end{figure}

\begin{theorem}\label{thmbounded1}
Let $G$ be a $(K_r \boxminus rP_1, 2P_2)$-free graph for $r \geq 3$.
Then $\cutmim_G(X, \overline{X}) < \max\{6,r\}$ for every $X \subseteq V(G)$.
In particular, $\mimw(G) < \max\{6,r\}$.
\end{theorem}

\begin{proof}
Let $k = \max\{6,r\}$ and let $(T,\delta)$ be a branch decomposition of $G$.
Towards a contradiction, suppose that there exists $X \subseteq V(G)$ such that $G[X,\overline{X}]$ has an induced matching of size at least $k$.
Let $X' = \{x_1,x_2,\dotsc,x_{k}\} \subseteq X$ and $Y' = \{y_1,y_2,\dotsc,y_{k}\} \subseteq \overline{X}$ such that $x_iy_i$ is an edge of the induced matching for each $i \in \{1,2,\dotsc,k\}$.

First, observe that for any distinct $i,j \in \{1,2,\dotsc,k\}$, either $x_ix_j$ or $y_iy_j$ is an edge, otherwise $G[\{x_i,x_j,y_i,y_j\}] \cong 2P_2$.
We claim that $X'$ or $Y'$ contains a clique of size $3$.
Since $|X'| = k \ge 6 = R(3,3)$, 
the set $X'$ contains either a clique on $3$ vertices, or an independent set on $3$ vertices.
So we may assume that $X'$ contains an independent set on $3$ vertices, $\{x_i,x_j,x_\ell\}$ say.
Then $\{y_i, y_j, y_\ell\}$ is a clique of size $3$ contained in $Y'$, proving the claim.

Without loss of generality, we may now assume that $X'$ contains a clique of size $3$.
Suppose $X'$ is not a clique.
Then there exist distinct $i,j \in \{1,2,\dotsc,k\}$ such that $x_i$ is not adjacent to $x_j$.
Now $y_iy_j$ is an edge, since $G$ is $2P_2$-free.
Let $X''$ be a maximum-sized clique contained in $X'$, so $|X''| \ge 3$.
Note that $\{x_i,x_j\} \nsubseteq X''$, since $X''$ is a clique, so we may assume that $x_j \notin X''$.
As any pair in $X'' \setminus \{x_i\}$ induces an edge that is anticomplete to the edge $y_iy_j$, 
we see that $G$ contains an induced $2P_2$, a contradiction.
We deduce that $X'$ is a clique of size $k$.
Now, since $G$ is $(K_r \boxminus rP_1)$-free, there exist distinct $i,j \in \{1,2,\dotsc,k\}$ such that $y_iy_j$ is an edge.
Note that since $k \ge 6$, there exist distinct $s,t \in \{1,2,\dotsc,k\} \setminus \{i,j\}$.
But now $x_sx_t$ is anticomplete to $y_iy_j$, contradicting that $G$ is $2P_2$-free.
\end{proof}
The class of $(K_r \boxminus rP_1, 2P_2)$-free graphs for $r \in \{1,2\}$ is a subclass of $P_4$-free graphs, and thus has bounded clique-width and mim-width.
However, for $r\ge 3$, the class of $(K_r \boxminus rP_1, 2P_2)$-free graphs has unbounded clique-width~\cite[Theorem~4.18]{DJP19}, whereas \cref{thmbounded1} shows it has bounded mim-width.
In particular, $(\net, 2P_2)$-free graphs and $(\bull, 2P_2)$-free graphs have bounded mim-width but unbounded clique-width.

In our next two results, we present two other new classes of bounded mim-width.

\begin{theorem}\label{thmbounded2}
Let $G$ be a $(K_r \boxminus P_1, tP_2)$-free graph for $r \geq 1$ and
$t \geq 1$.
Then $\cutmim_G(X, \overline{X}) < R(r,R(r,t))$ for every $X \subseteq V(G)$.
In particular, $\mimw(G) < R(r,R(r,t))$.
\end{theorem}

\begin{proof}
Let $k = R(r,R(r,t))$ and let $(T,\delta)$ be a branch decomposition of $G$.
Towards a contradiction, suppose that there exists $X \subseteq V(G)$ such that $G[X,\overline{X}]$ has an induced matching of size at least $k$. Let $X' = \{x_1,x_2,\dotsc,x_{k}\} \subseteq X$ and $Y' = \{y_1,y_2,\dotsc,y_{k}\} \subseteq \overline{X}$ such that $x_iy_i$ is an edge of the induced matching for each $i \in \{1,2,\dotsc,k\}$.

Since $|X'| = k = R(r,R(r,t))$, the set $X'$ contains either a clique of size $r$, or an independent set of size $R(r,t)$. Suppose there is some $J \subseteq \{1,2,\dotsc,k\}$ such that $X_J=\{x_i : i \in J\}$ is a clique of size $r$. Then, for an arbitrarily chosen $j \in J$, the vertices $X_J \cup \{y_j\}$ induce a $K_r \boxminus P_1$, a contradiction. So $X'$ contains an independent set of size $R(r,t)$. Let $I \subseteq \{1,2,\dotsc,k\}$ such that $X_I=\{x_i : i \in I\}$ is an independent set of size $R(r,t)$, and consider the set $Y_I=\{y_i : i \in I\}$. Since $|Y_I| = R(r,t)$, the set $Y_I$ either contains a clique of size $r$, or an independent set of size $t$. In the former case, $G$ contains an induced $K_r \boxminus P_1$, while in the latter case, $G$ contains an induced $tP_2$, a contradiction.
\end{proof}

\begin{theorem}\label{thmbounded3}
Let $G$ be a $(K_r \boxminus K_r, sP_1 + P_2)$-free graph for $r \geq 1$
and $s \geq 0$.
Then $\cutmim_G(X, \overline{X}) < R(R(r,s+1),s+1)$ for every $X \subseteq V(G)$.
In particular, $\mimw(G) < R(R(r,s+1),s+1)$.
\end{theorem}

\begin{proof}
Let $k = R(R(r,s+1),s+1)$ and let $(T,\delta)$ be a branch decomposition of $G$.
Towards a contradiction, suppose that there exists $X \subseteq V(G)$ such that $G[X,\overline{X}]$ has an induced matching of size at least $k$. Let $X' = \{x_1,x_2,\dotsc,x_{k}\} \subseteq X$ and $Y' = \{y_1,y_2,\dotsc,y_{k}\} \subseteq \overline{X}$ such that $x_iy_i$ is an edge of the induced matching for each $i \in \{1,2,\dotsc,k\}$.

Since $|X'| = k = R(R(r,s+1),s+1)$, the set $X'$ contains either a clique of size $R(r,s+1)$, or an independent set of size $s+1$. But the latter implies that $G$ has an induced $sP_1+P_2$ subgraph, a contradiction. So $X'$ contains a clique of size $R(r,s+1)$. Let $I \subseteq \{1,2,\dotsc,k\}$ such that $X_I=\{x_i : i \in I\}$ is an clique of size $R(r,s+1)$, and consider the set $Y_I=\{y_i : i \in I\}$. Since $|Y_I| = R(r,s+1)$, the set $Y_I$ either contains a clique of size $r$, or an independent set of size $s+1$. In the former case, $G$ contains an induced $K_r \boxminus K_r$, while in the latter case, $G$ contains an induced $sP_1 + P_2$, a contradiction.
\end{proof}

Note that $(K_r \boxminus P_1, tP_2)$-free graphs have
unbounded clique-width if and only if $r \ge 3$,$t \ge 3$, or $r \ge 4$, $t \ge 2$ \cite[Theorem~4.18]{DJP19}.
Note also that $(K_r \boxminus K_r, sP_1+P_2)$-free graphs have unbounded clique-width if and only if $r=2$, $s \ge 3$, or $r \ge 3$, $s \ge 2$~\cite[Theorem~4.18]{DJP19}.

\medskip
\noindent
Our final results of the section are used to resolve the remaining cases where $|V(H_1)| + |V(H_2)| \le 8$.\footnote{In Corollary~\ref{c-8} in Section~\ref{s-soa} we prove that we determined all pairs $(H_1,H_2)$ with 
$|V(H_1)| + |V(H_2)| \le 8$ for which the mim-width of $(H_1,H_2)$-free graphs is bounded and quickly computable.}
For these results, we employ the following approach.
Suppose we wish to show that the class of $(H'_1,H'_2)$-free graphs is bounded, where $H'_1 \ssi H_1$ for one of the pairs $(H_1,H_2)$ appearing in \cref{thmbounded1,thmbounded2,thmbounded3}.
  If $G$ is a $H_2$-free graph in the class, then we can compute a branch decomposition of constant mim-width by one of \cref{thmbounded1,thmbounded2,thmbounded3}.
  So it remains only to show that we can compute a branch decomposition of constant mim-width for $(H'_1,H'_2)$-free graphs having an induced subgraph isomorphic to $H_2$.
When $H_1' = 2P_2$ and $H_2' = K_{1,3}$, we exploit the structure of $(2P_2,K_{1,3})$-free graphs having an induced $K_3 \boxminus 3P_1$ to prove Lemma~\ref{t-2p2claw-lemma}. Then, by combining this lemma with \cref{thmbounded1}, we obtain \cref{t-2p2claw}.
  Similarly, when $H_1' = 2P_1+P_2$ and $H_2' = \bowtie$ (see \cref{boundedegfigs}), we use \cref{t-3p1bowtie-lemma,thmbounded3} to obtain \cref{t-3p1bowtie}.
  
  For the proofs of \cref{t-2p2claw-lemma,t-3p1bowtie-lemma}, we require the following definition.
For an integer $l \geq 1$, an {\em $l$-caterpillar} is a subcubic tree $T$ on
$2l$ vertices with $V(T) = \{ s_1,\ldots,s_l,t_1,\ldots,t_l \}$,
such that $E(T) = \{ s_it_i \;:\; 1 \leq i \leq l \} \cup \{ s_is_{i+1} \;:\; 1 \leq i \leq l-1 \}$.
Note that we label the leaves of an $l$-caterpillar $t_1,t_2,\dotsc,t_l$, in this order.
See Figure~\ref{f-cater} for an example.

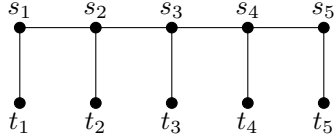
\begin{figure}[h]
\begin{center}
\begin{tikzpicture}[scale=1]
\draw (-2,-0.5)--(-2,0.5)--(2,0.5)--(2,-0.5) (1,-0.5)--(1,0.5) (0,-0.5)--(0,0.5) (-1,-0.5)--(-1,0.5);
\draw[fill=black] (2,0.5) circle [radius=2pt] (1,0.5) circle [radius=2pt] (0,0.5) circle [radius=2pt]
(-1,0.5) circle [radius=2pt] (-2,0.5) circle [radius=2pt] (2,-0.5) circle [radius=2pt] (1,-0.5) circle [radius=2pt]
(0,-0.5) circle [radius=2pt] (-1,-0.5) circle [radius=2pt] (-2,-0.5) circle [radius=2pt] ;
\node[above] at (-2,0.5) {$s_1$}; \node[above] at (-1,0.5) {$s_2$};
\node[above] at (0,0.5) {$s_3$}; \node[above] at (1,0.5) {$s_4$};
\node[above] at (2,0.5) {$s_5$}; \node[below] at (-2,-0.5) {$t_1$};
\node[below] at (-1,-0.5) {$t_2$}; \node[below] at (0,-0.5) {$t_3$};
\node[below] at (1,-0.5) {$t_4$}; \node[below] at (2,-0.5) {$t_5$};
\end{tikzpicture}
\caption{The $5$-caterpillar.}\label{f-cater}
\end{center}
\end{figure}

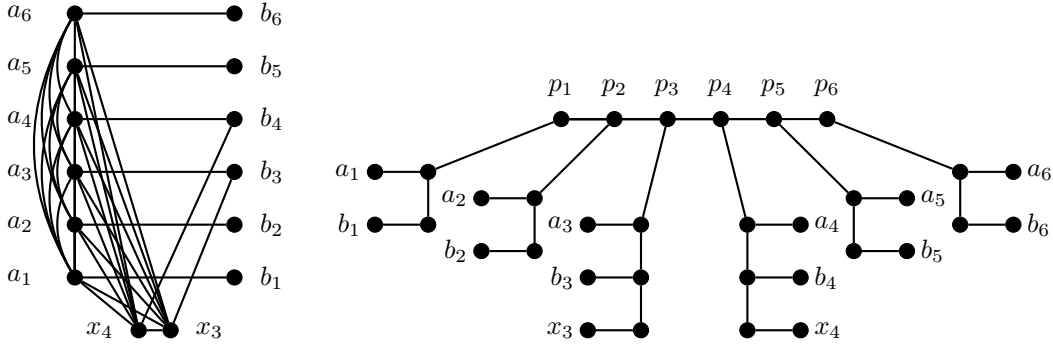
\begin{figure}[ht]
  \centering
  \begin{tikzpicture}[scale=0.7]
    \node[inner sep=0pt, minimum size=0.2cm,circle,fill,draw] (a1) at (0,0) {};
    \node[inner sep=0pt, minimum size=0.2cm,circle,fill,draw] (a2) at (0,1) {};
    \node[inner sep=0pt, minimum size=0.2cm,circle,fill,draw] (a3) at (0,2) {};
    \node[inner sep=0pt, minimum size=0.2cm,circle,fill,draw] (a4) at (0,3) {};
    \node[inner sep=0pt, minimum size=0.2cm,circle,fill,draw] (a5) at (0,4) {};
    \node[inner sep=0pt, minimum size=0.2cm,circle,fill,draw] (a6) at (0,5) {};
    \node[inner sep=0pt, minimum size=0.2cm,circle,fill,draw] (b1) at (3,0) {};
    \node[inner sep=0pt, minimum size=0.2cm,circle,fill,draw] (b2) at (3,1) {};
    \node[inner sep=0pt, minimum size=0.2cm,circle,fill,draw] (b3) at (3,2) {};
    \node[inner sep=0pt, minimum size=0.2cm,circle,fill,draw] (b4) at (3,3) {};
    \node[inner sep=0pt, minimum size=0.2cm,circle,fill,draw] (b5) at (3,4) {};
    \node[inner sep=0pt, minimum size=0.2cm,circle,fill,draw] (b6) at (3,5) {};
    \foreach\i in {1,...,6}{
      \draw (a\i) edge[thick] (b\i);
      \node[left=0.9em of a\i] {$a_{\i}$};
      \node[right=0.3em of b\i] {$b_{\i}$};
    }
    \draw (a1) edge[thick] (a2) edge[thick] (a3) edge[thick] (a4) edge[thick] (a5) edge[thick] (a6);
    \draw (a1) edge[thick,bend left] (a3);
    \draw (a1) edge[thick,bend left] (a4);
    \draw (a1) edge[thick,bend left] (a5);
    \draw (a1) edge[thick,bend left] (a6);
    \draw (a2) edge[thick,bend left] (a4);
    \draw (a2) edge[thick,bend left] (a5);
    \draw (a2) edge[thick,bend left] (a6);
    \draw (a3) edge[thick,bend left] (a5);
    \draw (a3) edge[thick,bend left] (a6);
    \draw (a4) edge[thick,bend left] (a6);
    \node[inner sep=0pt, minimum size=0.2cm,circle,fill,draw] (B2) at (1.2,-1) {};
    \node[inner sep=0pt, minimum size=0.2cm,circle,fill,draw] (B1) at (1.8,-1) {};
    \node[left=0.1cm of B2] {$x_4$};
    \node[right=0.1cm of B1] {$x_3$};
    \draw (B1) edge[thick] (B2);
    \draw (B2) edge[thick] (b4);
    \draw (B1) edge[thick] (b3);
    \draw (B1) edge[thick] (a1);
    \draw (B1) edge[thick] (a2);
    \draw (B1) edge[thick] (a3);
    \draw (B1) edge[thick] (a4);
    \draw (B1) edge[thick] (a5);
    \draw (B1) edge[thick] (a6);
    \draw (B2) edge[thick] (a1);
    \draw (B2) edge[thick] (a2);
    \draw (B2) edge[thick] (a3);
    \draw (B2) edge[thick] (a4);
    \draw (B2) edge[thick] (a5);
    \draw (B2) edge[thick] (a6);
  \end{tikzpicture}
 \hspace{0.6em}
  \begin{tikzpicture}[scale=0.7]
    \node[inner sep=0pt, minimum size=0.2cm,circle,fill,draw] (p1) at (0,1) {};
    \node[inner sep=0pt, minimum size=0.2cm,circle,fill,draw] (p2) at (1,1) {};
    \node[inner sep=0pt, minimum size=0.2cm,circle,fill,draw] (p3) at (2,1) {};
    \node[inner sep=0pt, minimum size=0.2cm,circle,fill,draw] (p4) at (3,1) {};
    \node[inner sep=0pt, minimum size=0.2cm,circle,fill,draw] (p5) at (4,1) {};
    \node[inner sep=0pt, minimum size=0.2cm,circle,fill,draw] (p6) at (5,1) {};
    \node[inner sep=0pt, minimum size=0.2cm,circle,fill,draw] (t1) at (-2.5,0) {};
    \node[inner sep=0pt, minimum size=0.2cm,circle,fill,draw] (tt1) at (-2.5,-1) {};
    \node[inner sep=0pt, minimum size=0.2cm,circle,fill,draw] (t2) at (-0.5,-0.5) {};
    \node[inner sep=0pt, minimum size=0.2cm,circle,fill,draw] (tt2) at (-0.5,-1.5) {};
    \node[inner sep=0pt, minimum size=0.2cm,circle,fill,draw] (t3) at (1.5,-1) {};
    \node[inner sep=0pt, minimum size=0.2cm,circle,fill,draw] (tt3) at (1.5,-2) {};
    \node[inner sep=0pt, minimum size=0.2cm,circle,fill,draw] (ttt3) at (1.5,-3) {};
    \node[inner sep=0pt, minimum size=0.2cm,circle,fill,draw] (t4) at (3.5,-1) {};
    \node[inner sep=0pt, minimum size=0.2cm,circle,fill,draw] (tt4) at (3.5,-2) {};
    \node[inner sep=0pt, minimum size=0.2cm,circle,fill,draw] (ttt4) at (3.5,-3) {};
    \node[inner sep=0pt, minimum size=0.2cm,circle,fill,draw] (t5) at (5.5,-0.5) {};
    \node[inner sep=0pt, minimum size=0.2cm,circle,fill,draw] (tt5) at (5.5,-1.5) {};
    \node[inner sep=0pt, minimum size=0.2cm,circle,fill,draw] (t6) at (7.5,0) {};
    \node[inner sep=0pt, minimum size=0.2cm,circle,fill,draw] (tt6) at (7.5,-1) {};
    \draw (p1) edge[thick] (p2) edge[thick] (p3) edge[thick] (p4) edge[thick] (p5) edge[thick] (p6);
    \foreach\i in {1,...,6}{
      \draw (p\i) edge[thick] (t\i);
      \draw (t\i) edge[thick] (tt\i);
      }
      \draw (tt3) edge[thick] (ttt3);
      \draw (tt4) edge[thick] (ttt4);
      \node[inner sep=0pt, minimum size=0.2cm,circle,fill,draw] (l1) at ({-3.5+2*(1-1)},0) {};
      \node[inner sep=0pt, minimum size=0.2cm,circle,fill,draw] (ll1) at ({-3.5+2*(1-1)},-1) {};
      \node[inner sep=0pt, minimum size=0.2cm,circle,fill,draw] (l2) at ({-3.5+2*(2-1)},-0.5) {};
      \node[inner sep=0pt, minimum size=0.2cm,circle,fill,draw] (ll2) at ({-3.5+2*(2-1)},-1.5) {};
      \node[inner sep=0pt, minimum size=0.2cm,circle,fill,draw] (l3) at ({-3.5+2*(3-1)},-1) {};
      \node[inner sep=0pt, minimum size=0.2cm,circle,fill,draw] (ll3) at ({-3.5+2*(3-1)},-2) {};

        \node[inner sep=0pt, minimum size=0.2cm,circle,fill,draw] (l4) at ({-1.5+2*(4-1)},-1) {};
        \node[inner sep=0pt, minimum size=0.2cm,circle,fill,draw] (ll4) at ({-1.5+2*(4-1)},-2) {};
        \node[inner sep=0pt, minimum size=0.2cm,circle,fill,draw] (l5) at ({-1.5+2*(5-1)},-0.5) {};
        \node[inner sep=0pt, minimum size=0.2cm,circle,fill,draw] (ll5) at ({-1.5+2*(5-1)},-1.5) {};
        \node[inner sep=0pt, minimum size=0.2cm,circle,fill,draw] (l6) at ({-1.5+2*(6-1)},0) {};
        \node[inner sep=0pt, minimum size=0.2cm,circle,fill,draw] (ll6) at ({-1.5+2*(6-1)},-1) {};

      \foreach\i in {1,...,6}{
      \draw (t\i) edge[thick] (l\i);
      \draw (tt\i) edge[thick] (ll\i);}
        \node[inner sep=0pt, minimum size=0.2cm,circle,fill,draw] (lll3) at ({-3.5+2*(2)},-3) {};
        \node[inner sep=0pt, minimum size=0.2cm,circle,fill,draw] (lll4) at ({-1.5+2*(3)},-3) {};
        \draw (ttt3) edge[thick] (lll3);
        \draw (ttt4) edge[thick] (lll4);
        \foreach\i in {1,...,6}{
          \node[above=0.2em of p\i] {$p_{\i}$};
          }
          \foreach\i in {1,2,3}{
            \node[left=-0.05cm of l\i] {$a_{\i}$};
            \node[left=-0.05cm of ll\i] {$b_{\i}$};
            }
          \foreach\i in {4,5,6}{
            \node[right=-0.05cm of l\i] {$a_{\i}$};
            \node[right=-0.05cm of ll\i] {$b_{\i}$};
            }
            \node[left=-0.05cm of lll3] {$x_3$};
            \node[right=-0.05cm of lll4] {$x_4$};

  \end{tikzpicture}
  \caption{On the left, a $(2P_2, K_{1,3})$-free graph $G$, and on the right
    the branch decomposition $(T,\delta)$ of $G$ as constructed in the proof of Lemma~\ref{t-2p2claw-lemma}.
  }
  \label{f-2p2claw}
\end{figure}

\begin{lemma}\label{t-2p2claw-lemma}
Let $G$ be a connected $(2P_2,K_{1,3})$-free graph.
Given $X \subseteq V(G)$ such that $G[X] \cong K_r \boxminus rP_1$ for some $r \ge 3$, where $X$ is maximal, we can construct, in $O(n)$ time, a branch decomposition $(T, \delta)$ of $G$ such that $\mimw_G(T,\delta) = 1$.
\end{lemma}

\begin{proof}
  Let $A = \{ a_1,\ldots,a_r \}$ and $B = \{ b_1,\ldots,b_r \}$ such that $A$ is a clique, $B$ is an independent set, and $(A,B)$ is a partition of $X$. 
  Note that $G[X] \cong K_r \boxminus rP_1$, but for every $S \subseteq V(G) \setminus X$, we have that $G[X \cup S] \not\cong K_{r'} \boxminus r'P_1$ for each integer $r' > r$.  We assume that $a_ib_i \in E(G)$ for each $i \in \{1,\ldots,r\}$.
  Let $N_1$ be the set of vertices from $V(G) \setminus X$ that
  have a neighbour in $X$, and 
  let $N_2 = V(G) \setminus (X \cup N_1)$.

  Let $v \in N_1$. Suppose that $N(v) \cap B = \varnothing$.
  Since $G$ is connected, $v$ has a neighbour in $A$; by symmetry, we
  may assume that $va_1 \in E(G)$.
  Let $i \in \{2,\ldots,r\}$, and suppose that $va_i \not\in E(G)$.
  But then $G[\{ a_1,b_1,v,a_i\}] \cong K_{1,3}$, a contradiction.
  Therefore $N(v) \cap X = A$.
  Suppose now that $N(v) \cap B \neq \varnothing$; without loss of
  generality we may
  assume that $vb_1 \in E(G)$.
  If $v$ is complete to $B$, then any three vertices of $B$ together
  with $v$ induces a $K_{1,3}$, a contradiction.
  Therefore, without loss of generality we assume that $vb_2 \not\in
  E(G)$.
  Since $G$ is $2P_2$-free, $va_2 \in E(G)$.
  Now suppose that $va_i \not\in E(G)$ for some $i \in \{1,\ldots,r\}
  \setminus \{2\}$.
  But then $G[\{a_2,b_2,v,a_i\}] \cong K_{1,3}$, a contradiction.
  Therefore $v$ is complete to $A$.
  Now suppose that $|N(v) \cap B| \geq 2$; without loss of generality
  we may assume that $b_1, b_3 \in N(v)$.
  Recall that $b_2 \not\in N(v)$.
  But then $G[\{v,b_1,b_3,a_2\}] \cong K_{1,3}$, a contradiction.
  Therefore $N(v) \cap B = \{b_1\}$.
  Hence, for every vertex $v \in N_1$,
  either $N(v) \cap X = A$ or $N(v) \cap X = A \cup
  \{b\}$ for some $b \in B$.

  Suppose that there exist vertices $v, v' \in N_1$ such that $vv'
  \not\in E(G)$.
  Since vertices of $N_1$ have at most one neighbour in $B$, we may
  assume without loss of generality that $b_1 \not\in N(v) \cup N(v')$.
  But then $G[\{a_1,b_1,v,v'\}] \cong K_{1,3}$, a contradiction.
  Therefore $vv' \in E(G)$, and hence $N_1$ is a clique.

  We now prove that $N_2 = \varnothing$.
  Towards a contradiction, suppose that there exists a vertex $w \in N_2$.
  Since $G$ is connected, there exists a vertex $v \in N(w) \cap N_1$.
  By what we have already proved, either $N(v) \cap X = A$ or
  $N(v) \cap X = A \cup \{b\}$ for some $b \in B$.
  Suppose that $N(v) \cap B \neq \varnothing$; without loss of generality,
  we may assume that $N(v) \cap B = \{b_1\}$.
  But then $G[\{v,b_1,w,a_2\}] \cong K_{1,3}$, a contradiction.
  Therefore $v$ is anticomplete to $B$.
  It now follows that $G[X \cup \{v,w\}] \cong K_{r+1} \boxminus
  (r+1)P_1$, contradicting the maximality of $X$.
  Therefore $N_2 = \varnothing$.

  For $i \in \{1,\ldots,r\}$, let $B_i$ denote the set of vertices
  from $N_1$ that are adjacent to $b_i$ and let $B_0$ denote the set of
  vertices from $N_1$ that have no neighbour in $B$.
  Note that $(A,B,B_0,B_1,\ldots,B_r)$ is a partition of $V(G)$ (into possibly empty sets), and we can construct this partition in $O(n)$ time.
  Consider the branch decomposition $(T, \delta)$ of $G$ defined as
  follows; see also Figure~\ref{f-2p2claw}.
  For each $i \in \{1,\ldots,r\}$,
  let $T_i$ be a $(|B_i|+2)$-caterpillar and let $t_i$ be a vertex of
  $T_i$ of degree 2.
  If $B_0 \neq \varnothing$, let $T_0$ be a $|B_0|$-caterpillar and $t_0$ a vertex of $T_0$ of
  degree 2, or of degree 1 if $|B_0|=1$.
  Let $P = p_1,\ldots,p_r$ be a path on $r$ vertices.
  Let $T'$ be the tree with
  $V(T') = V(P) \cup \bigcup_{i=1}^r V(T_i)$ and
  $E(T') = E(P) \cup \bigcup_{i=1}^r E(T_i) \cup \{ t_ip_i \;:\; 1 \leq i \leq r \}$.
  If $B_0 = \varnothing$ then let $T=T'$, and otherwise let
  $T$ be the tree obtained from $T'$ by adding an additional vertex
  $p_{r+1}$ together with all vertices of $V(T_0)$, and adding edges
  $p_rp_{r+1}$ and $p_{r+1}t_0$ together with all edges of $T_0$.
  Finally, let $\delta$ be any bijection from $V(G)$ to the leaves of
  $T$ such that for all $i \in \{1,\ldots,r\}$ and for all $v \in V(G)$,
  $\delta(v) \in V(T_i)$ if $v \in \{a_i,b_i\} \cup B_i$,
  and $\delta(v) \in V(T_0)$ if $v \in B_0$.

  We now prove that $\mimw_G(T,\delta) = 1$.
   Let $e$ be an edge of $T$ and let $M$ be a maximum induced matching of
  $G[A_e, \overline{A_e}]$.
  We begin by claiming that at most one edge of $M$ has one endpoint in $B$ and the other
  in $A \cup N_1$.
  On the contrary, suppose without loss of generality that $b_1x$ and
  $b_2y$ are distinct edges of $M$, where $b_1,b_2 \in B \cap A_e$ and
  $x,y \in (A \cup N_1) \cap \overline{A_e}$.
  Observe that if $x \in N_1$ (respectively $y \in N_1$), then $x \in B_1$
  (respectively $y \in B_2$); and if $x \in A$ (respectively $y \in A$), then $x=a_1$
  (respectively $y=a_2$).
  Since $b_1, b_2 \in A_e$, we have that $e \not\in E(T_1) \cup E(T_2)
  \cup \{p_1p_2, p_1t_1, p_2t_2\}$, and therefore $\{a_1,a_2\} \cup B_1 \cup B_2
  \subseteq A_e$.
  But $N(b_1) \cup N(b_2) \subseteq \{a_1,a_2\} \cup B_1 \cup B_2$, a
  contradiction.
  Therefore at most one edge of $M$ has one endpoint in $B$ and the other in
  $A \cup N_1$.
  Since $A \cup N_1$ is a clique, at most one edge of $M$ has both endpoints
  in $A \cup N_1$, and since $B$ is an independent set, no edge of $M$ has
  both endpoints in $B$.
  Suppose that $|M| \geq 2$.
  Then $M = \{ uv, xy \}$, where, without loss of generality,
  $u,x \in A_e$, $u, v, x \in A \cup N_1$ and $y \in B$.
  But since $A \cup N_1$ is a clique, $xv$ is an edge, contradicting $M$
  being an induced matching.
  Therefore $|M| \leq 1$ and hence $\mimw_G(T, \delta) = 1$, as required.
\end{proof}

\begin{theorem}\label{t-2p2claw}
Let $G$ be a  $(2P_2,K_{1,3})$-free graph.
Then $\mimw(G) < 6$, and
one can construct, in polynomial time, a branch decomposition
$(T,\delta)$ of $G$ with $\mimw_G(T, \delta) < 6$.
\end{theorem}

\begin{proof}
  If $G$ is not connected, we may consider each component in turn, by \cref{l-2con}.
  If $G$ is $(K_3 \boxminus 3P_1)$-free, then $\mimw(G) < 6$ by \cref{thmbounded1}.  On the other hand, if $G$ has an induced subgraph isomorphic to $K_3 \boxminus 3P_1$, then $\mimw(G) = 1$ by \cref{t-2p2claw-lemma}.

We now show how to compute a branch decomposition $(T,\delta)$ of $G$, with $\mimw_G(T,\delta) < 6$, in polynomial time.
Consider the following algorithm, which takes as input a connected $(2P_2,K_{1,3})$-free graph $G$.

  \begin{enumerate}[Step~1]
	\item Enumerate all subsets $S \subseteq V(G)$
	  such that $|S| = 6$ and check whether $G[S] \cong K_3 \boxminus 3P_1$.
	  If no such set $S$ exists, then return an arbitrary branch decomposition of $G$.
	\item Let $S \subseteq V(G)$ such that $G[S] \cong K_3 \boxminus
	  3P_1$ and let $(A,B)$ be a partition of $S$ such that $A$ is a
	  clique and $B$ is an independent set.
	\item Set $E = E(G) \setminus E(G[S])$.
	  While $E \neq \varnothing$:
	  \begin{itemize}
	    \item Choose an edge $e \in E$.
	    \item If one endpoint of $e$ (say $a$) is complete to $A$ and
	      anticomplete to $B$, and the other endpoint of $e$ (say $b$) is
	      anticomplete to $A \cup B$, then set $A \gets A \cup
	      \{a\}$ and $B \gets B \cup \{b\}$.
	    \item Set $E \gets E \setminus \{e\}$.
	  \end{itemize}
	\item Using \cref{t-2p2claw-lemma}, with $X = A \cup B$, compute a branch decomposition $(T,\delta)$ of $G$ and return it.
  \end{enumerate}

      It is easily checked that Steps 1--4 of this algorithm can be
      performed in polynomial time.
      If the algorithm returns a branch decomposition in Step 1,
      then by Theorem~\ref{thmbounded1} it has mim-width less than 6.
      Otherwise, the branch decomposition has mim-width 1 by \cref{t-2p2claw-lemma}.
\end{proof}

\begin{lemma}\label{t-3p1bowtie-lemma}
Let $G$ be a $(2P_1+P_2,\bowtie)$-free graph.
Given $X \subseteq V(G)$ such that $G[X] \cong K_r \boxminus K_r$ for some $r \ge 5$, where $X$ is maximal, we can construct, in $O(n)$ time, a branch decomposition $(T, \delta)$ of $G$ such that $\mimw_G(T,\delta) = 2$.
\end{lemma}

\begin{proof}
  Let $A = \{a_1,\ldots,a_r\}$ and $B = \{b_1,\ldots,b_r\}$ be cliques that partition $X$, with $a_ib_i \in E(G)$ for all $i \in \{1,\dotsc,r\}$.
  Let $N_1$ be the set of vertices of $V(G) \setminus X$ with a neighbour in $X$.
  Suppose there exists a vertex $v \in V(G) \setminus (X \cup N_1)$.
  Then $G[\{v,a_1,b_2,b_3\}] \cong 2P_1+P_2$, a contradiction.
  So $X \cup N_1 = V(G)$.

  We claim each vertex in $N_1$ is either complete or anticomplete to $A$.
  Suppose $v \in N_1$ has a neighbour and a non-neighbour in $A$.  Without loss of generality, let $a_r$ be the neighbour and let $a_1$ be the non-neighbour.
  If there is a pair of distinct vertices $b_i,b_j$ non-adjacent to $v$ for $i,j \in \{2,3,\dotsc,r\}$, then $G[\{v,a_1,b_i,b_j\}] \cong 2P_1+P_2$.
  So $v$ has at most one non-neighbour in $\{b_2,b_3,\dotsc,b_r\}$.
  In particular, as $r \ge 5$, we may assume without loss of generality that $b_3$ and $b_4$ are neighbours of $v$.
  If $v$ is adjacent to $a_2$, then $G[\{a_2,a_r,v,b_3,b_4\}] \cong \bowtie$, a contradiction.
  So $a_2$ is a non-neighbour of $v$.
  Now, if $b_r$ is adjacent to $v$, then $G[\{a_1,a_2,a_r,v,b_r\}] \cong \bowtie$; whereas if $b_r$ is non-adjacent to $v$, then $G[\{a_1,a_2,v,b_r\}] \cong 2P_1+P_2$.
  From this contradiction, we deduce that $v$ is either complete or anticomplete to $A$.
  By symmetry, each $v \in N_1$ is complete or anticomplete to $B$.

  If $v \in N_1$ is complete to both $A$ and $B$, then $G[\{a_1,a_2,v,b_3,b_4\}] \cong \bowtie$, a contradiction.
  If $v \in N_1$ is anticomplete to both $A$ and $B$, then $G[\{a_1,a_2,v,b_3\}] \cong 2P_1+P_2$, a contradiction.
  So each vertex in $N_1$ is either complete to $A$ and anticomplete to $B$, or complete to $B$ and anticomplete to $A$.  Call these two sets $A'$ and $B'$ respectively.  If a vertex $a \in A'$ has a neighbour $b \in B'$, then $G[X \cup \{a,b\}] \cong K_{r+1} \boxminus K_{r+1}$, contradicting the maximality of $X$.  So $A'$ and $B'$ are anticomplete.
  Moreover, if $a,a' \in A'$ are distinct and non-adjacent, then $G[\{a,a',b_1,b_2\}] \cong 2P_1+P_2$, a contradiction. So $A' \cup A$ and, similarly, $B' \cup B$ are cliques.

  Now let $a'_1,a'_2,\dotsc,a'_{|A'|}$ be an arbitrary ordering of $A'$, and let $b'_1,b'_2,\dotsc,b'_{|B'|}$ be an arbitrary ordering of $B'$.
  Let $(T,\delta)$ be the branch decomposition with linear ordering \[(a'_1, a'_2, \dotsc, a'_{|A'|}, a_1, b_1, a_2, b_2, \dotsc, a_r, b_r, b'_1, b'_2, \dotsc, b'_{|B'|});\] that is, let $T$ be a $|V(G)|$-caterpillar where $\delta$ respects this ordering, so $\delta(a'_1)=t_1$, $\delta(a'_2)=t_2$, \ldots, $\delta(b'_{|B'|})=t_{|V(G)|}$.
  Note that, given $X$, we can find $A$ and $B$, together with the labelling of $a_i$'s and $b_i$'s, as well as $A'$ and $B'$, in $O(n)$ time, so we can compute $(T,\delta)$ in $O(n)$ time.
  We claim that $\mimw_G(T,\delta) = 2$.
  Let $e \in E(T)$ and consider the corresponding cut $(A_e, \overline{A_e})$.
  First, observe that when $A_e = A' \cup \{a_1,b_1\}$, the graph $G[A_e, \overline{A_e}]$ has an induced matching of size 2, with edges $a_1a_2$ and $b_1b_2$, so $\mimw_G(T,\delta) \ge 2$.

  Let $M$ be an induced matching in $G[A_e,\overline{A_e}]$.  Let $V(M)$ denote the vertices incident to an edge of $M$.
  Suppose $V(M) \cap A_e$ contains at least two vertices of $A$.
  Then there exist $i,j\in \{1,2,\dotsc,r\}$ such that $a_i,a_j \in V(M) \cap A_e$, with $i < j$.  Observe that $b_i \in A_e$, since $b_i$ is between $a_i$ and $a_j$ in the linear ordering, and $a_i,a_j \in A_e$.
  Let $v,v' \in \overline{A_e}$ such that $a_iv,a_jv' \in M$.
  If $v \in A \cup A'$, then $a_jv$ is an edge of $G$, so $M$ is not induced.  Moreover, $v \notin B'$, since $B'$ is anticomplete to $A$.  So $v \in B$, and hence $v = b_i$.  But then $v \in A_e$, a contradiction.
  So $|V(M) \cap A_e \cap A| \le 1$.
  Similarly, $V(M) \cap A_e$ contains at most one vertex of $B$.

  Now suppose $V(M) \cap A_e$ contains a vertex $a' \in A'$.
  Suppose $a'v \in M$, where $a' \in A' \cap A_e$ and $v \in \overline{A_e}$.
  Then $v \in A \cup A'$, since $A'$ is anticomplete to $B \cup B'$.
  Hence $a'$ is the only vertex of $A \cup A'$ in $A_e \cap V(M)$, for otherwise $v$ has two neighbours in $V(M) \cap A_e$.
  So $|V(M) \cap A_e \cap (A' \cup A)| \le 1$.
  Similarly, $V(M) \cap A_e$ contains at most one vertex of $B' \cup B$.
  So $|M| \le 2$, and hence $\mimw_G(T,\delta) = 2$.
\end{proof}

\begin{theorem}\label{t-3p1bowtie}
Let $G$ be a  $(2P_1+P_2,\bowtie)$-free graph.
Then $\mimw(G) < R(14,3)$, and
one can construct, in polynomial time, a branch decomposition
$(T,\delta)$ of $G$ with $\mimw_G(T, \delta) < R(14,3)$.
\end{theorem}

\begin{proof}
  If $G$ is $(K_5 \boxminus K_5)$-free, then $\mimw(G) < R(R(5,3),3)=R(14,3)$ by \cref{thmbounded3}.  On the other hand, if $G$ has an induced subgraph isomorphic to $K_5 \boxminus K_5$, then $\mimw(G) = 2$ by \cref{t-3p1bowtie-lemma}.

We now show how to compute a branch decomposition $(T,\delta)$ of $G$, with $\mimw_G(T,\delta) < R(14,3)$, in polynomial time.
Consider the following algorithm, which takes as input a connected $(2P_1+P_2,\bowtie)$-free graph $G$.

  \begin{enumerate}[Step~1]
	\item Enumerate all subsets $S \subseteq V(G)$
	  such that $|S| = 10$ and check whether $G[S] \cong K_5 \boxminus K_5$.
	  If no such set $S$ exists, then return an arbitrary branch decomposition of $G$.
	\item Let $S \subseteq V(G)$ such that $G[S] \cong K_5 \boxminus K_5$ and let $(A,B)$ be a partition of $S$ such that $A$ is a
	  clique and $B$ is an independent set.
	\item Set $E = E(G) \setminus E(G[S])$.
	  While $E \neq \varnothing$:
	  \begin{itemize}
	    \item Choose an edge $e \in E$.
	    \item If one endpoint of $e$ (say $a$) is complete to $A$ and
	      anticomplete to $B$, and the other endpoint of $e$ (say $b$) is
	      complete to $B$ and anticomplete to $A$, then set $A \gets A \cup
	      \{a\}$ and $B \gets B \cup \{b\}$.
	    \item Set $E \gets E \setminus \{e\}$.
	  \end{itemize}
	\item Using \cref{t-3p1bowtie-lemma}, with $X = A \cup B$, compute a branch decomposition $(T,\delta)$ of $G$ and return it.
  \end{enumerate}

      It is easily checked that Steps 1--4 of this algorithm can be
      performed in polynomial time.
      If the algorithm returns a branch decomposition in Step 1,
      then by \cref{thmbounded3} it has mim-width less than R(14,3).
      Otherwise, the branch decomposition has mim-width 2 by \cref{t-3p1bowtie-lemma}.
\end{proof}

\section{New Unbounded Cases}
\label{sec:unbounded}

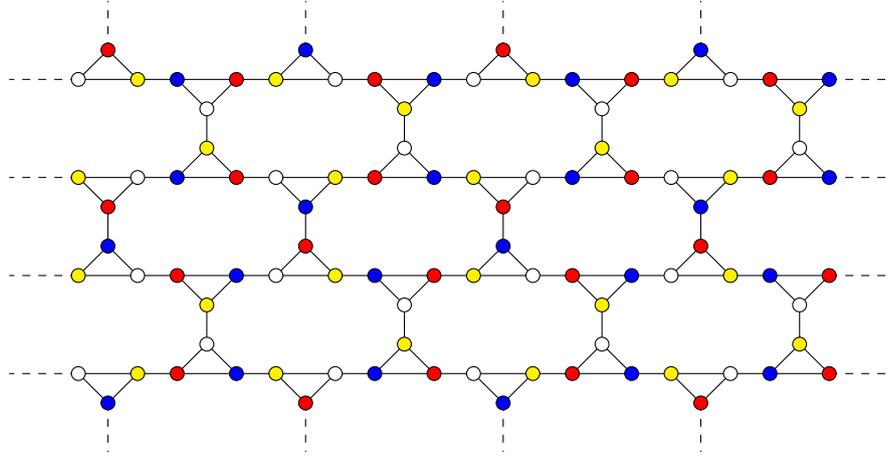
\begin{figure}
    \begin{center}
        \begin{tikzpicture}[scale=1.3, rotate=90]
            \draw
            (2,-3.3)--(2,4.3) (1,-3.3)--(1,4.3) (0,-3.3)--(0,4.3) (-1,-3.3)--(-1,4.3)
            (-1,-3.3)--(-0.7,-3)--(-1,-2.7) (-1,-1.3)--(-0.7,-1)--(-1,-0.7) (-1,1.3)--(-0.7,1)--(-1,0.7) (-1,3.3)--(-0.7,3)--(-1,2.7)
            (-1,-2.3)--(-1.3,-2)--(-1,-1.7) (-1,0.3)--(-1.3,0)--(-1,-0.3) (-1,2.3)--(-1.3,2)--(-1,1.7) (-1,4.3)--(-1.3,4)--(-1,3.7) 
            (0,-2.3)--(0.3,-2)--(0,-1.7) (0,0.3)--(0.3,0)--(0,-0.3) (0,2.3)--(0.3,2)--(0,1.7) (0,4.3)--(0.3,4)--(0,3.7)
            (0,-3.3)--(-0.3,-3)--(0,-2.7) (0,-1.3)--(-0.3,-1)--(0,-0.7)  (0,1.3)--(-0.3,1)--(0,0.7)(0,3.3)--(-0.3,3)--(0,2.7)
            (1,-3.3)--(1.3,-3)--(1,-2.7) (1,-1.3)--(1.3,-1)--(1,-0.7) (1,1.3)--(1.3,1)--(1,0.7) (1,3.3)--(1.3,3)--(1,2.7)
            (1,-2.3)--(0.7,-2)--(1,-1.7) (1,0.3)--(0.7,0)--(1,-0.3) (1,2.3)--(0.7,2)--(1,1.7) (1,4.3)--(0.7,4)--(1,3.7)
            (2,-2.3)--(2.3,-2)--(2,-1.7) (2,0.3)--(2.3,0)--(2,-0.3) (2,2.3)--(2.3,2)--(2,1.7) (2,4.3)--(2.3,4)--(2,3.7) 
            (2,-3.3)--(1.7,-3)--(2,-2.7) (2,-1.3)--(1.7,-1)--(2,-0.7) (2,1.3)--(1.7,1)--(2,0.7) (2,3.3)--(1.7,3)--(2,2.7)
            (-0.7,-3)--(-0.3,-3) (-0.7,-1)--(-0.3,-1) (-0.7,1)--(-0.3,1) (-0.7,3)--(-0.3,3)
            (1.7,-3)--(1.3,-3) (1.7,-1)--(1.3,-1) (1.7,1)--(1.3,1) (1.7,3)--(1.3,3)
            (0.7,-2)--(0.3,-2) (0.7,0)--(0.3,0) (0.7,2)--(0.3,2) (0.7,4)--(0.3,4);
            \draw[dashed]%
           (-1,-3.3)--(-1,-4) (0,-3.3)--(0,-4) (1,-3.3)--(1,-4) (2,-3.3)--(2,-4)
            (-1,4.3)--(-1,5) (0,4.3)--(0,5) (1,4.3)--(1,5) (2,4.3)--(2,5)
            (-1.3,4)--(-1.8,4) (-1.3,2)--(-1.8,2) (-1.3,0)--(-1.8,0) (-1.3,-2)--(-1.8,-2)
            (2.3,4)--(2.8,4) (2.3,2)--(2.8,2) (2.3,0)--(2.8,0) (2.3,-2)--(2.8,-2);
            \draw[fill=red]
            (-1,-3.3) circle [radius=2pt] (-1,-0.7) circle [radius=2pt] (-1,0.7) circle [radius=2pt] (-1,3.3) circle [radius=2pt] (0,-3.3) circle [radius=2pt]  
            (0,-0.7) circle [radius=2pt]  (0,0.7) circle [radius=2pt]  (0,3.3) circle [radius=2pt] (1,-2.7) circle [radius=2pt] (1,-1.3) circle [radius=2pt]
            (1,1.3) circle [radius=2pt] (1,2.7) circle [radius=2pt] (2,-2.7) circle [radius=2pt] (2,-1.3) circle [radius=2pt] (2,1.3) circle [radius=2pt] 
            (2,2.7) circle [radius=2pt] (-1.3,-2) circle [radius=2pt] (-1.3,2) circle [radius=2pt] (0.3,-2) circle [radius=2pt] (0.3,2) circle [radius=2pt] 
            (0.7,0) circle [radius=2pt] (0.7,4) circle [radius=2pt] (2.3,0) circle [radius=2pt] (2.3,4) circle [radius=2pt];
            \draw[fill=blue]
            (-1,-2.7) circle [radius=2pt] (-1,-1.3) circle [radius=2pt] (-1,1.3) circle [radius=2pt](-1,2.7) circle [radius=2pt] (0,-2.7) circle [radius=2pt] 
            (0,-1.3) circle [radius=2pt]  (0,1.3) circle [radius=2pt] (0,2.7) circle [radius=2pt] (1,-3.3) circle [radius=2pt] (1,-0.7) circle [radius=2pt]  
            (1,0.7) circle [radius=2pt]  (1,3.3) circle [radius=2pt]  (2,-3.3) circle [radius=2pt] (2,-0.7) circle [radius=2pt] (2,0.7) circle [radius=2pt] 
            (2,3.3) circle [radius=2pt] (-1.3,0) circle [radius=2pt] (-1.3,4) circle [radius=2pt] (0.3,0) circle [radius=2pt] (0.3,4) circle [radius=2pt] 
            (0.7,-2) circle [radius=2pt] (0.7,2) circle [radius=2pt] (2.3,-2) circle [radius=2pt] (2.3,2) circle [radius=2pt];
            \draw[fill=yellow]
            (-1,-1.7) circle [radius=2pt] (-1,-0.3) circle [radius=2pt] (-1,2.3) circle [radius=2pt] (-1,3.7) circle [radius=2pt] (0,-2.3) circle [radius=2pt]  
            (0,0.3) circle [radius=2pt]  (0,1.7) circle [radius=2pt] (0,4.3) circle [radius=2pt] (1,-2.3) circle [radius=2pt] (1,0.3) circle [radius=2pt]  
            (1,1.7) circle [radius=2pt]  (1,4.3) circle [radius=2pt] (2,-1.7) circle [radius=2pt] (2,-0.3) circle [radius=2pt] (2,2.3) circle [radius=2pt] 
            (2,3.7) circle [radius=2pt] (-0.7,-3) circle [radius=2pt](-0.7,1) circle [radius=2pt] (-0.3,-1) circle [radius=2pt]
            (-0.3,3) circle [radius=2pt] (1.7,-3) circle [radius=2pt] (1.7,1) circle [radius=2pt] (1.3,-1) circle [radius=2pt](1.3,3) circle [radius=2pt];
            \draw[fill=white] 
            (-1,-2.3) circle [radius=2pt] (-1,0.3) circle [radius=2pt] (-1,1.7) circle [radius=2pt] (-1,4.3) circle [radius=2pt] (0,-1.7) circle [radius=2pt] 
            (0,-0.3) circle [radius=2pt] (0,2.3) circle [radius=2pt] (0,3.7) circle [radius=2pt] (1,-1.7) circle [radius=2pt] (1,-0.3) circle [radius=2pt] 
            (1,2.3) circle [radius=2pt] (1,3.7) circle [radius=2pt] (2,-2.3) circle [radius=2pt] (2,0.3) circle [radius=2pt] (2,1.7) circle [radius=2pt] 
            (2,4.3) circle [radius=2pt] (-0.7,-1) circle [radius=2pt] (-0.7,3) circle [radius=2pt](-0.3,-3) circle [radius=2pt](-0.3,1) circle [radius=2pt]
            (1.7,-1) circle [radius=2pt] (1.7,3) circle [radius=2pt](1.3,-3) circle [radius=2pt](1.3,1) circle [radius=2pt];

        \end{tikzpicture}
        \caption{A particular $4$-colouring of a net-wall, used in the proof of \cref{diamond-5p1}.}\label{fig-netwall4col}
    \end{center}
\end{figure}

We present a number of graph classes of unbounded mim-width, starting with following two theorems.

\begin{theorem}
    \label{diamond-5p1}
    The class of $(\diamond, 5P_1)$-free graphs has unbounded mim-width.
\end{theorem}

\begin{proof}
For every integer $k$, we will construct a $(\diamond, 5P_1)$-free graph~$G$ such that $\mimw(G) > k$. By \cref{netwalls}, for any integer $k$ there exists a net-wall $W$ such that $\mimw(W) > 4k$. We partition the vertex set $V(W)$ into four colour classes $(V_1,V_2,V_3,V_4)$ as illustrated in \cref{fig-netwall4col}. Observe that, for each $i \in \{1,2,3,4\}$, the set $V_i$ is independent, and no two distinct vertices $v,v' \in V_i$ have a common neighbour; that is, $N_{W}(v) \cap N_{W}(v') = \varnothing$.

Let $G$ be the graph obtained from $W$ by making each of $V_1$, $V_2$, $V_3$ and $V_4$ into a clique. By \cref{kpartite}, $\mimw(G) \ge \mimw(W)/4 > k$. Since any set of five vertices of $G$ contains at least two vertices in one of $V_1$, $V_2$, $V_3$, and $V_4$, and each of these four sets is a clique, $G$ is $5P_1$-free.

It remains to show that $G$ is $\diamond$-free. First, observe that if $G[X] \cong K_3$ for some $X \subseteq V(G)$ with $|X \cap V_i| \ge 2$ for some $i \in \{1,2,3,4\}$, then, since no two vertices in $V_i$ have a common neighbour in $W$, it follows that $X \subseteq V_i$. Now, towards a contradiction, suppose $G[Y] \cong \diamond$ for some $Y \subseteq V(G)$. Then $Y$ is the union of two sets $X'$ and $X''$ that induce triangles in $G$, and $|X' \cap X''| = 2$. Since $W$ is $\diamond$-free, we may assume that $W[X']$ is not a triangle. Then $X'$ contains at least two vertices of $V_i$ for some $i \in \{1,2,3,4\}$. By the earlier observation, $X' \subseteq V_i$. Since $|X' \cap X''| = 2$, we then have $|X'' \cap V_i| \ge 2$, so $X'' \subseteq V_i$,  and hence $Y \subseteq V_i$. But this implies that $Y$ is a clique in $G$; a contradiction. So $G$ is $\diamond$-free.
\end{proof}

\begin{figure}
    \begin{center}
        \begin{tikzpicture}[scale=1.3, rotate=90]
            \draw
            (2,-3.3)--(2,4.3) (1,-3.3)--(1,4.3) (0,-3.3)--(0,4.3) (-1,-3.3)--(-1,4.3)
            (-1,-3.3)--(-0.7,-3)--(-1,-2.7) (-1,-1.3)--(-0.7,-1)--(-1,-0.7) (-1,1.3)--(-0.7,1)--(-1,0.7) (-1,3.3)--(-0.7,3)--(-1,2.7)
            (-1,-2.3)--(-1.3,-2)--(-1,-1.7) (-1,0.3)--(-1.3,0)--(-1,-0.3) (-1,2.3)--(-1.3,2)--(-1,1.7) (-1,4.3)--(-1.3,4)--(-1,3.7) 
            (0,-2.3)--(0.3,-2)--(0,-1.7) (0,0.3)--(0.3,0)--(0,-0.3) (0,2.3)--(0.3,2)--(0,1.7) (0,4.3)--(0.3,4)--(0,3.7)
            (0,-3.3)--(-0.3,-3)--(0,-2.7) (0,-1.3)--(-0.3,-1)--(0,-0.7)  (0,1.3)--(-0.3,1)--(0,0.7)(0,3.3)--(-0.3,3)--(0,2.7)
            (1,-3.3)--(1.3,-3)--(1,-2.7) (1,-1.3)--(1.3,-1)--(1,-0.7) (1,1.3)--(1.3,1)--(1,0.7) (1,3.3)--(1.3,3)--(1,2.7)
            (1,-2.3)--(0.7,-2)--(1,-1.7) (1,0.3)--(0.7,0)--(1,-0.3) (1,2.3)--(0.7,2)--(1,1.7) (1,4.3)--(0.7,4)--(1,3.7)
            (2,-2.3)--(2.3,-2)--(2,-1.7) (2,0.3)--(2.3,0)--(2,-0.3) (2,2.3)--(2.3,2)--(2,1.7) (2,4.3)--(2.3,4)--(2,3.7) 
            (2,-3.3)--(1.7,-3)--(2,-2.7) (2,-1.3)--(1.7,-1)--(2,-0.7) (2,1.3)--(1.7,1)--(2,0.7) (2,3.3)--(1.7,3)--(2,2.7)
            (-0.7,-3)--(-0.3,-3) (-0.7,-1)--(-0.3,-1) (-0.7,1)--(-0.3,1) (-0.7,3)--(-0.3,3)
            (1.7,-3)--(1.3,-3) (1.7,-1)--(1.3,-1) (1.7,1)--(1.3,1) (1.7,3)--(1.3,3)
            (0.7,-2)--(0.3,-2) (0.7,0)--(0.3,0) (0.7,2)--(0.3,2) (0.7,4)--(0.3,4);
            \draw[dashed]
            (-1,-3.3)--(-1,-4) (0,-3.3)--(0,-4) (1,-3.3)--(1,-4) (2,-3.3)--(2,-4)
            (-1,4.3)--(-1,5) (0,4.3)--(0,5) (1,4.3)--(1,5) (2,4.3)--(2,5)
            (-1.3,4)--(-1.8,4) (-1.3,2)--(-1.8,2) (-1.3,0)--(-1.8,0) (-1.3,-2)--(-1.8,-2)
            (2.3,4)--(2.8,4) (2.3,2)--(2.8,2) (2.3,0)--(2.8,0) (2.3,-2)--(2.8,-2);
            \draw[fill=red]
            (-1,-1.7) circle [radius=2pt] (-1,0.3) circle [radius=2pt] (-1,2.3) circle [radius=2pt](-1,4.3) circle [radius=2pt] (0,-2.7) circle [radius=2pt] (0,-0.7) circle [radius=2pt] 
            (0,1.3) circle [radius=2pt] (0,3.3) circle [radius=2pt] (1,-1.7) circle [radius=2pt] (1,0.3) circle [radius=2pt] (1,2.3) circle [radius=2pt] (1,4.3) circle [radius=2pt] 
            (2,-2.7) circle [radius=2pt] (2,-0.7) circle [radius=2pt] (2,1.3) circle [radius=2pt] (2,3.3) circle [radius=2pt]  (0.3,-2) circle [radius=2pt] (0.3,0) circle [radius=2pt] 
            (0.3,2) circle [radius=2pt] (0.3,4) circle [radius=2pt] (2.3,-2) circle [radius=2pt] (2.3,0) circle [radius=2pt] (2.3,2) circle [radius=2pt] (2.3,4) circle [radius=2pt] 
            (-0.7,-3) circle [radius=2pt] (-0.7,-1) circle [radius=2pt] (-0.7,1) circle [radius=2pt] (-0.7,3) circle [radius=2pt] 
            (1.3,-3) circle [radius=2pt] (1.3,-1) circle [radius=2pt] (1.3,1) circle [radius=2pt] (1.3,3) circle [radius=2pt];
            \draw[fill=blue]
            (-1,-2.7) circle [radius=2pt] (-1,-0.7) circle [radius=2pt] (-1,1.3) circle [radius=2pt] (-1,3.3) circle [radius=2pt] (0,-1.7) circle [radius=2pt] (0,0.3) circle [radius=2pt] 
            (0,2.3) circle [radius=2pt] (0,4.3) circle [radius=2pt] (1,-2.7) circle [radius=2pt] (1,-0.7) circle [radius=2pt] (1,1.3) circle [radius=2pt] (1,3.3) circle [radius=2pt] 
            (2,-1.7) circle [radius=2pt] (2,0.3) circle [radius=2pt] (2,2.3) circle [radius=2pt] (2,4.3) circle [radius=2pt] (-1.3,-2) circle [radius=2pt] (-1.3,0) circle [radius=2pt] 
            (-1.3,2) circle [radius=2pt] (-1.3,4) circle [radius=2pt] (0.7,-2) circle [radius=2pt] (0.7,0) circle [radius=2pt] (0.7,2) circle [radius=2pt] (0.7,4) circle [radius=2pt] 
            (-0.3,-3) circle [radius=2pt] (-0.3,-1) circle [radius=2pt] (-0.3,1) circle [radius=2pt] (-0.3,3) circle [radius=2pt] 
            (1.7,-3) circle [radius=2pt] (1.7,-1) circle [radius=2pt] (1.7,1) circle [radius=2pt] (1.7,3) circle [radius=2pt];
            \draw[fill=white] 
            (-1,-3.3) circle [radius=2pt] (-1,-2.3) circle [radius=2pt] (-1,-1.3) circle [radius=2pt] (-1,-0.3) circle [radius=2pt] (-1,0.7) circle [radius=2pt](-1,1.7) circle [radius=2pt] 
            (-1,2.7) circle [radius=2pt] (-1,3.7) circle [radius=2pt] (0,-3.3) circle [radius=2pt] (0,-2.3) circle [radius=2pt] (0,-1.3) circle [radius=2pt] (0,-0.3) circle [radius=2pt] 
            (0,0.7) circle [radius=2pt] (0,1.7) circle [radius=2pt] (0,2.7) circle [radius=2pt] (0,3.7) circle [radius=2pt] (1,-3.3) circle [radius=2pt](1,-2.3) circle [radius=2pt] 
            (1,-1.3) circle [radius=2pt] (1,-0.3) circle [radius=2pt] (1,0.7) circle [radius=2pt] (1,1.7) circle [radius=2pt] (1,2.7) circle [radius=2pt] (1,3.7) circle [radius=2pt] 
            (2,-3.3) circle [radius=2pt] (2,-2.3) circle [radius=2pt] (2,-1.3) circle [radius=2pt] (2,-0.3) circle [radius=2pt] (2,0.7) circle [radius=2pt] (2,1.7) circle [radius=2pt] 
            (2,2.7) circle [radius=2pt] (2,3.7) circle [radius=2pt];
        \end{tikzpicture}
        \caption{The $3$-colouring of a net-wall used in the proof of Theorem~\ref{k5minus-4p1}.}\label{fig-netwall3col}
    \end{center}
\end{figure}
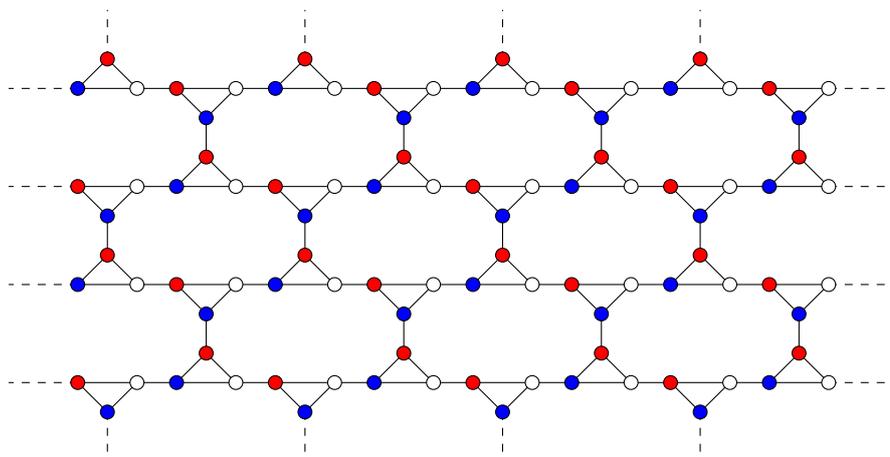

\begin{theorem}
    \label{k5minus-4p1}
    The class of $(4P_1,\overline{3P_1+P_2},\overline{P_1+2P_2})$-free graphs has unbounded mim-width.  
\end{theorem}

\begin{proof}
   For every integer $k$, we will construct a $(4P_1,\overline{3P_1+P_2},\overline{P_1+2P_2})$-free graph~$G$ such that $\mimw(G) > k$.
    By \cref{netwalls}, for any integer $k$ there exists a net-wall $W$ such that $\mimw(W) > 3k$.
    We partition the vertex set $V(W)$ into three colour classes $(V_1,V_2,V_3)$ such that $V_i$ is an independent set for each $i \in \{1,2,3\}$ as illustrated in \cref{fig-netwall3col}.
    Since $W$ has maximum degree~$3$ and each vertex belongs to a triangle, a vertex has at most two neighbours in each colour class; that is, for each $i \in \{1,2,3\}$ and $v \in V_i$, we have $|N(v) \cap V_j| \le 2$ for $j \in \{1,2,3\}$.
    Note that these colour classes are chosen to satisfy the following properties.
    Firstly, $W$ does not contain a bichromatic induced $P_5$; that is, if $W[X] \cong P_5$ for some $X \subseteq V(P_5)$, then $X \cap V_i \neq \varnothing$ for each $i \in \{1,2,3\}$.
    Secondly, if $W[X] \cong \bull$, then $|X \cap V_i| \le 2$ for each $i \in \{1,2,3\}$.

    Let $G$ be the graph obtained from $W$ by making each of $V_1$, $V_2$, and $V_3$ into a clique.
    By \cref{kpartite}, $\mimw(G) \ge \mimw(W)/3 > k$.
    As any set of 4 vertices of $G$ contains at least two vertices in one of the cliques $V_1$, $V_2$, or $V_3$, we deduce that $G$ is $4P_1$-free.

    We now show that $G$ is $\left(\overline{3P_1+P_2}\right)$-free.
    To the contrary, suppose $G[X] \cong \overline{3P_1+P_2}$ for some $X \subseteq V(G)$.
    Then $X$ is not contained in $V_i$ for any $i \in \{1,2,3\}$.
    Moreover, $|X \cap V_i| \le 2$ for each $i\in \{1,2,3\}$, for otherwise there is a vertex with at least three neighbours in a different colour class.
    So, assume without loss of generality that 
    $X \cap V_1 = \{v_1,v'_1\}$, $X \cap V_2 = \{v_2,v'_2\}$, and $X \cap V_3 = \{v_3\}$.
    Then at least two of $\{v_1,v_2,v_3\}$, $\{v_1',v_2,v_3\}$, $\{v_1,v'_2,v_3\}$, $\{v_1',v_2',v_3\}$ induce triangles in $G$.
    These triangles consist of one vertex in each colour class, so they correspond to induced triangles in $W$.
    This is contradictory, as $W$ has no two distinct triangles that share a vertex.

    It remains to show that $G$ is $\left(\overline{P_1+2P_2}\right)$-free.
    Towards a contradiction, suppose $G[X] \cong \overline{P_1+2P_2}$.
    Note that $G[X]$ has a dominating vertex $h$.
    Without loss of generality, let $h \in V_1$.
    Since $h$ has degree $4$ in $G[X]$, we have $|X \cap V_1| \ge 2$.
    In fact, as $W$ has no cycles of length 4, no two vertices in $V_2 \cup V_3$ share a pair of common neighbours in $V_1$, so $|X \cap V_1| \ge 3$.
    Since $G[X]$ is $K_4$-free, we have $|X \cap V_1| = 3$.
    Let $X \cap V_1 = \{x,x',h\}$ and $X \setminus V_1 = \{y,z\}$.
    We may assume
    without loss of generality that $y \in V_2$.
    Now there is a $5$-vertex path $xyhzx'$ in $W$, up to the labels of $x$ and $x'$.
    If $z \in V_2$, then the four edges of this path are the only edges in $G[\{x,y,h,z,x'\}]$ where the two endpoints are in different colour classes, so $W[\{x,y,h,z,x'\}] \cong P_5$.
    Since $W$ has no bichromatic induced $P_5$, we deduce that $z \in V_3$.
    But then $W[X] \cong \bull$ and $|X \cap V_1| =3$, a contradiction.
\end{proof}

Next we use the  construction of a chordal bipartite graph $G'$ from a graph $G$, given in~\cite{BCM15}\footnote{Alternatively, we could take a wall, which has bipartition classes $A$ and $B$; 2-subdivide all of its edges; and make $A$ complete to $B$. The resulting graph has the same structure as $G'$ and can have arbitrarily large mim-width due to Theorem~\ref{walls} and Lemmas~\ref{subdivision} and~\ref{kpartite}.}.
Let $G=(V,E)$ be a graph.  We take two copies of $V$ labelled as follows: $X = \{x_v : v \in V\}$ and $Y = \{y_v : v \in V\}$.
To construct $G'$, start with a complete bipartite graph with vertex bipartition $(X,Y)$, and add, for each edge $e \in E$ with endpoints $u$ and $v$, two paths:
an $x_uy_v$-path $x_{u}q_et_ey_{v}$, and an $x_vy_u$-path $x_{v}q'_et'_ey_{u}$.
For convenience, we let $Q = \bigcup_{e \in E(G)} \{q_e, q'_e\}$ and $T = \bigcup_{e \in E(G)} \{t_e, t'_e\}$. Observe that $(X,Y,Q,T)$ partitions $V(G')$; see also Figure~\ref{f-gg}.

\begin{figure}[htb]
\begin{minipage}{0.48\textwidth}
  \centering
\begin{tikzpicture}[scale=0.8]
\draw (-3,2)--(-2.3,1)--(-2.3,-1)--(-1,-2)--(0.3,-1)--(0.3,1)--(1,2)--(1.7,1)--(1.7,-1)--(3,-2)
(-3,-2)--(-1.7,-1)--(-1.7,1)--(-1,2)--(-0.3,1)--(-0.3,-1)--(1,-2)--(2.3,-1)--(2.3,1)--(3,2)
(-3.2,2.3)--(-3.5,2.3)--(-3.5,1.7)--(-3.2,1.7) (-3.2,-2.3)--(-3.5,-2.3)--(-3.5,-1.7)--(-3.2,-1.7)
(-3.2,1.3)--(-3.5,1.3)--(-3.5,0.7)--(-3.2,0.7) (-3.2,-1.3)--(-3.5,-1.3)--(-3.5,-0.7)--(-3.2,-0.7);
\draw[fill=black] (-3,2) circle [radius=3pt] (-1,2) circle [radius=3pt] (1,2) circle [radius=3pt] (3,2) circle [radius=3pt] (-3,-2) circle [radius=3pt] 
(-1,-2) circle [radius=3pt] (1,-2) circle [radius=3pt] (3,-2) circle [radius=3pt] (-2.3,1) circle[radius=2pt] (-1.7,1) circle[radius=2pt] 
(-0.3,1) circle[radius=2pt] (0.3,1) circle[radius=2pt] (1.7,1) circle [radius=2pt] (2.3,1) circle[radius=2pt] (-2.3,-1) circle[radius=2pt] 
(-1.7,-1) circle[radius=2pt] (-0.3,-1) circle[radius=2pt] (0.3,-1) circle[radius=2pt] (1.7,-1) circle [radius=2pt] (2.3,-1) circle[radius=2pt];
\node[left] at (-3.5,2) {$Y$};
\node[left] at (-3.5,1) {$T$};
\node[left] at (-3.5,-1) {$Q$};
\node[left] at (-3.5,-2) {$X$};
\node[right] at (-3,2) {$y_a$};
\node[right] at (-1,2) {$y_b$};
\node[left] at (1,2) {$y_c$};
\node[left] at (3,2) {$y_d$};
\node[right] at (-3,-2) {$x_a$};
\node[right] at (-1,-2) {$x_b$};
\node[left] at (1,-2) {$x_c$};
\node[left] at (3,-2) {$x_d$};
\node[left] at (-2.3,1) {$t'_{ab}$};
\node[right] at (-1.7,1) {$t_{ab}$};
\node[left] at (-0.3,1) {$t'_{bc}$};
\node[right] at (0.3,1) {$t_{bc}$};
\node[left] at (1.7,1) {$t'_{cd}$};
\node[right] at (2.3,1) {$t_{cd}$};
\node[left] at (-2.3,-1) {$q'_{ab}$};
\node[right] at (-1.7,-1) {$q_{ab}$};
\node[left] at (-0.3,-1) {$q'_{bc}$};
\node[right] at (0.3,-1) {$q_{bc}$};
\node[left] at (1.7,-1) {$q'_{cd}$};
\node[right] at (2.3,-1) {$q_{cd}$};
\end{tikzpicture}
\end{minipage}
\begin{minipage}{0.48\textwidth}
 \centering
\begin{tikzpicture}[scale=0.8]
\draw (-3,2)--(-2.3,1)--(-2.3,-1)--(-1,-2)--(0.3,-1)--(0.3,1)--(1,2)--(1.7,1)--(1.7,-1)--(3,-2)
(-3,-2)--(-1.7,-1)--(-1.7,1)--(-1,2)--(-0.3,1)--(-0.3,-1)--(1,-2)--(2.3,-1)--(2.3,1)--(3,2)
(3.2,2.3)--(3.5,2.3)--(3.5,1.7)--(3.2,1.7) (3.2,-2.3)--(3.5,-2.3)--(3.5,-1.7)--(3.2,-1.7)
(3.2,-0.3)--(3.5,-0.3)--(3.5,0.3)--(3.2,0.3);
\draw[fill=black] (-3,2) circle [radius=3pt] (-1,2) circle [radius=3pt] (1,2) circle [radius=3pt] (3,2) circle [radius=3pt] 
(-3,-2) circle [radius=3pt] (-1,-2) circle [radius=3pt] (1,-2) circle [radius=3pt] (3,-2) circle [radius=3pt] 
(-2.3,0) circle[radius=2pt] (-1.7,0) circle[radius=2pt] (-0.3,0) circle[radius=2pt] (0.3,0) circle[radius=2pt] 
(1.7,0) circle [radius=2pt] (2.3,0) circle[radius=2pt];
\node[right] at (3.5,2) {$Y$};
\node[right] at (3.5,0) {$Z$};
\node[right] at (3.5,-2) {$X$};
\node[right] at (-3,2) {$y_a$};
\node[right] at (-1,2) {$y_b$};
\node[left] at (1,2) {$y_c$};
\node[left] at (3,2) {$y_d$};
\node[right] at (-3,-2) {$x_a$};
\node[right] at (-1,-2) {$x_b$};
\node[left] at (1,-2) {$x_c$};
\node[left] at (3,-2) {$x_d$};
\node[above left] at (-2.3,0) {$z'_{ab}$};
\node[below right] at (-1.7,0) {$z_{ab}$};
\node[above left] at (-0.3,0) {$z'_{bc}$};
\node[below right] at (0.3,0) {$z_{bc}$};
\node[above left] at (1.7,0) {$z'_{cd}$};
\node[below right] at (2.3,0) {$z_{cd}$};
\end{tikzpicture}
\end{minipage}
\caption{The graphs $G'$ and $G''$, where we did not draw the edges between $X$ and $Y$.}\label{f-gg}
\end{figure}
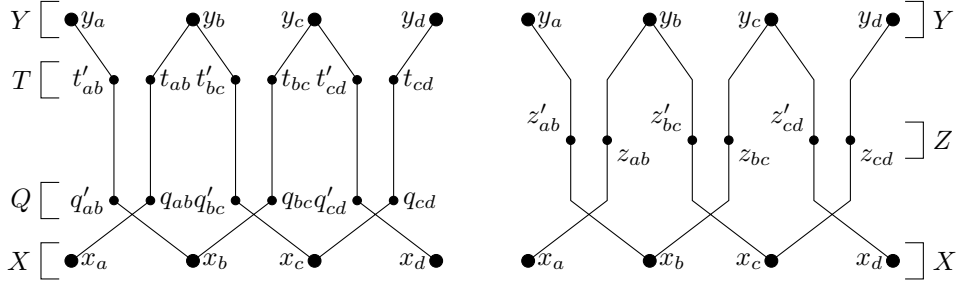

We need two lemmas. The  first one is due to Baron, Capelli and Mengel. 

\begin{lemma}[{\cite[Lemmas 15 and 16]{BCM15}}]
    \label{bcmconstruction}
    For any graph $G$, the graph $G'$ is chordal bipartite,
    Moreover, 
    if $G$ is bipartite, then $\mimw(G') \ge \tw(G) / 6$, where $\tw(G)$ denotes the treewidth of $G$.
\end{lemma}

\begin{lemma}\label{bcmconstrpaths}
    For any graph $G$, the chordal bipartite graph $G'$ is $(P_8, P_3+P_6, S_{1,1,5})$-free.
\end{lemma}

\begin{proof}
We label the vertices of $G'$ as described in the construction, so $(X, Y, Q, T)$ is a partition of $V(G')$. We first claim that if some $A \subseteq V(G')$ induces a path in $G'$, with $|A| \ge 6$, then $X \cap A$ and $Y \cap A$ are non-empty. Suppose $G'[A] \cong P_{|A|}$ and $Y \cap A = \varnothing$. In $G'[X \cup Q \cup T]$, each vertex in $T$ has degree~$1$, and each vertex in $Q$ has two neighbours: one in $X$ and one in $T$.  If a vertex of $T$ is in $A$, then it is an end of the path $G'[A]$; so $|T \cap A| \le 2$. If a vertex of $Q$ is in $A$, then either it is an end of the path $G'[A]$, or it is adjacent to a vertex of $T$ that is an end of the path $G'[A]$. So $|Q \cap A| \le 2$. Since $X$ is independent, $|A| \le 5$. The claim now follows by symmetry.

Now suppose some $A \subseteq V(G')$ induces a path in $G'$ where $A\cap X \neq \varnothing$ and $A\cap Y \neq \varnothing$.
    Since $G'[X \cup Y]$ is complete bipartite, we may also assume that $|X \cap A| \in \{1,2\}$ and $|Y \cap A| = 1$.
    For each vertex $v \in Q \cap A$ (respectively, $v \in T \cap A$), either $v$ is the end of the path $G'[A]$, or $v$ has a neighbour in $X \cap A$ (respectively, $Y \cap A$).
    Suppose $|(Q \cup T) \cap A| \ge 5$.
    Let $A'$ be the vertices in $(Q \cup T) \cap A$ that are not ends of the path $G'[A]$.
    Then $|A'| \ge 3$, and each vertex in $A'$ has a neighbour in $(X \cup Y) \cap A$.
    Since $A\cap X$ and $A \cap Y$ are non-empty, no two vertices in $(Q \cup T) \cap A$ share a neighbour in $(X \cup Y) \cap A$.
    So $|N_{G'[A]}(A') \cap (X \cup Y)| \ge 3$, implying $|X \cap A|=2$.  
    However, then the vertex in the singleton set $Y \cap A$ has degree~$3$ in $G'[A]$, a contradiction.
    So $|(Q \cup T) \cap A| < 5$, and $|A| < 8$. 
    It now follows that $G'$ is $P_8$-free.

    Next we suppose, for some $F \subseteq V(G')$, that $G'[F]$ is a linear forest, one component of which is a $P_6$.  Let $A \subseteq F$ such that $G'[A] \cong P_6$.
    By the foregoing claim, $X \cap A$ and $Y \cap A$ are non-empty.
    Since $G'[X \cup Y]$ is complete bipartite, it follows that $F \setminus A \subseteq Q \cup T$.
    Hence $G'[F \setminus A] \cong sP_1+tP_2$ for some $s,t \ge 0$, implying $G'$ is $(P_3+P_6)$-free.

    Finally, suppose $G'[S] \cong S_{1,1,5}$ for some $S \subseteq V(G)$.  Let $A \subseteq S$ such that $G'[A] \cong P_7$.
    By the foregoing, $X \cap A$ and $Y \cap A$ are non-empty, and $|(Q \cup T) \cap A| < 5$.  Hence $\{|X \cap A|, |Y \cap A|\} = \{1,2\}$.
    Observe now that both ends of the path $G'[A]$ are in either $Q$ or $T$, and the vertices of the path adjacent to the ends are in either $T$ or $Q$, respectively.
    But then some vertex in $T$ or $Q$ has degree~$3$ in $G'[S]$ and hence in $G'$, a contradiction.
    Hence $G'$ is $S_{1,1,5}$-free.
\end{proof}

\noindent
\Cref{bcmconstrpaths} is tight in the following sense: for some graph $G$, the graph $G'$ can contain, as an induced 
subgraph, $tP_2+P_7$ or $tP_5$ for any non-negative integer $t$, or $S_{2,2,4}$.

Theorem~\ref{cb-p8} now follows from \cref{bcmconstruction,bcmconstrpaths} and the fact that bipartite graphs can have arbitrarily large treewidth (see, e.g.,~\cite{Va12}). 
We use \cref{kpartite} to obtain \cref{gem-4p1,diamond-2p3}.

\begin{theorem}
    \label{cb-p8}
The class of chordal  bipartite $(P_8, P_3+P_6, S_{1,1,5})$-free graphs has unbounded mim-width. 
\end{theorem}

\begin{theorem}
  \label{gem-4p1}
  The class of $(4P_1, \gem, \overline{P_1+2P_2})$-free graphs has unbounded mim-width.  
\end{theorem}

\begin{proof}
  For every integer $k$, we will construct a $(4P_1,\gem)$-free graph~$G$ such that $\mimw(G) > k$.  Let $B$ be a bipartite graph with $\tw(B) > 24k$. Then, $\mimw(B') > 4k$ by \cref{bcmconstruction}. Observe that $B'$ is $4$-partite, where $V(B')$ has a partition $(X, Y, T, Q)$ into independent colour classes, using the labelling described in the construction. Let $G$ be the graph obtained from $B'$ by making $X$, $Y$, $T$, and $Q$ into cliques. By \cref{kpartite}, $\mimw(G) \ge \mimw(B')/4 > k$.

  Observe that $X \cup Y$, $T$, and $Q$ are cliques that partition $V(G)$, so $G$ is $4P_1$-free.
  Note also that each vertex in $Q$ has exactly one neighbour in $T$, exactly one neighbour in $X$, and no neighbours in $Y$.
  By symmetry, each vertex in $T$ has exactly one neighbour in $Q$, exactly one neighbour in $Y$, and no neighbours in $X$.
  In particular, each vertex in $Q \cup T$ has at most one neighbour in $X \cup Y$.
  It remains to show that $G$ is $(\gem, \overline{P_1+2P_2})$-free.

  Suppose $G[D] \cong \diamond$ for some $D \subseteq V(G)$.
  Since $X \cup Y$ is a clique, $|D \cap (X \cup Y)| \le 3$.
  In fact, $|D \cap (X \cup Y)| \le 1$, since each vertex in $Q \cup T$ has at most one neighbour in $X \cup Y$.
  Note also that $D \nsubseteq T \cup Q$, since a vertex in $T$ has at most one neighbour in $Q$ (and vice versa).
  It follows, without loss of generality, that $|D \cap Q| = 3$ and $|D \cap X| = 1$.

  Now suppose $G[D']$ is isomorphic to $\gem$ or $\overline{P_1+2P_2}$ for some $D' = D \cup z$ with $z \in V(G) \setminus D$.
  Note that a $\gem$ or a $\overline{P_1+2P_2}$ has a dominating vertex~$h$, and $h \in D \cap Q$.
  If $z \in X$, then $hz$ is not an edge, since the only neighbour of $h$ in $X$ is the vertex in $D \cap X$.
  If $z \in Y \cup T$, then $z$ has degree $1$ in $G[D']$.
  If $z \in Q$, then $G[D']$ contains a $K_4$.
  From this contradiction we deduce that $G$ is $(\gem, \overline{P_1+2P_2})$-free.
\end{proof}

\begin{theorem}
    \label{diamond-2p3}
The class of  $(\diamond, 2P_3)$-free graphs has unbounded mim-width.  
\end{theorem}

\begin{proof}
For every integer $k$, we will construct a $(\diamond, 2P_3)$-free graph~$G$ such that $\mimw(G) > k$.
    Let $B$ be a bipartite graph with $\tw(B) > 12k$.
    Then, $\mimw(B') > 2k$ by \cref{bcmconstruction}.
    Observe that $B'$ is bipartite, where $(X \cup T, Y \cup Q)$ is a bipartition of $V(B')$.
    Let $G$ be the graph obtained from $B'$ by making $X$ and $Y$ into cliques.
    By \cref{kpartite}, $\mimw(G) \ge \mimw(B')/2 > k$.

    Observe now that $X \cup Y$ is a clique of $G$.
    Moreover, $G$ can be obtained starting from $G[X \cup Y]$ by adding $3$-edge $xy$-paths for some $x \in X$ and $y \in Y$.
    It follows that each induced $P_3$ subgraph of $G$ contains some vertex of $X \cup Y$.
    Since $X \cup Y$ is a clique, any two disjoint induced $P_3$ subgraphs of $G$ have an edge between them.  So $G$ is $2P_3$-free.

    Finally, observe that for each induced $K_3$ subgraph of $G$ we have $V(K_3) \subseteq X \cup Y$.  Hence, if $G[A] \cong \diamond$ for some $A\subseteq V(G)$, then $A \subseteq X \cup Y$, but then $A$ is a clique, a contradiction. So $G$ is $\diamond$-free.
\end{proof}

We now describe the construction of a graph $G''$ from a graph $G=(V,E)$.
This construction is similar to the construction of $G'$; we adapt the approach taken by \cite{BCM15} to construct graphs with arbitrarily large mim-width. 
Take two copies of $V$ labelled as follows: $X = \{x_v : v \in V\}$ and $Y = \{y_v : v \in V\}$.
Construct a graph $G''$ on vertex set $X \cup Y \cup Z$ where $Z = \bigcup_{e \in E(G)} \{z_e, z'_e\}$.
Start with a complete bipartite graph with vertex bipartition $(X,Y)$, and add, for each edge $e \in E$ with endpoints $u$ and $v$, two paths $x_{u}z_ey_{v}$ and $x_{v}z'_ey_{u}$.
Observe that $G''$ is $3$-partite, with colour classes $(X,Y,Z)$; see also Figure~\ref{f-gg}.

The following lemma is proven by modifying the proof of \cref{bcmconstruction} given in \cite{BCM15}. Alternatively, we could take the 
$n \times n$ wall $W$, which has bipartition classes $A$ and $B$; $1$-subdivide each edge of $W$; and make $A$ complete to $B$. By applying Theorem~\ref{walls} and Lemmas~\ref{subdivision} and~\ref{kpartite}, we obtain a lower bound on the mim-width in terms of $n$.

\begin{lemma}
    \label{modified-bcmconstruction}
    If $G$ is a bipartite graph, then $\mimw(G'') \ge \tw(G) / 6$.
\end{lemma}

\begin{proof}
Let $G$ be a bipartite graph with vertex bipartition $(A,B)$, and let $(T'',\delta'')$ be an arbitrary branch decomposition of $G''$.  We will show that $\mimw_{G''}(T'',\delta'') \ge \tw(G)/6$.

We first construct a branch decomposition $(T,\delta)$ of $G$ such that $E(T) \subseteq E(T'')$, as follows. Let $T$ be the tree obtained from $T''$ by deleting the leaves $t \in V(T'')$ such that $\delta''(t) = x_v$ for some $v \in B$, or $\delta''(t) = y_u$ for some $u \in A$, or $\delta''(t) \in Q \cup T$. In the resulting tree $T$, for each leaf $t \in T$ we define $\delta(t) = v$ if $\delta''(t) = x_v$ for some $v \in A$; and $\delta(t) = u$ if $\delta''(t) = x_u$ for some $u \in B$.
  
Suppose $e \in E(T)$.  Recall that $(A_e,\overline{A_e})$ denotes the partition of $V(G)$ induced by the two components of $T \backslash e$, and let $(A''_e, \overline{A''_e})$ denote the partition of $V(G'')$ induced by the two components of $T'' \backslash e$. Let $uv$ be an edge in the cut $G[A_e,\overline{A_e}]$. Since $G$ is bipartite, we may assume $u \in A$ and $v \in B$. Then $x_u$ and $y_v$ are on different sides of the cut $G''[A''_e,\overline{A''_e}]$; we may assume that $x_u \in A''_e$ and $y_v \in \overline{A''_e}$. Since there is a path $x_uz_{uv}y_v$ in $G''$, either the edge $x_uz_{uv}$ or the edge $z_{uv}y_v$ is in $G''[A''_e,\overline{A''_e}]$.

Let $M$ be a matching of $G[A_e,\overline{A_e}]$. We obtain a matching $M'$ of $G[A''_e,\overline{A''_e}]$ of size $|M|$ as follows: for each edge $uv$ in $M$, choose the edge $x_uz_{uv}$ or $z_{uv}y_v$ that is in $G[A''_e,\overline{A''_e}]$. We partition $M'$ into $(M'_X,M'_Y)$ where $M'_X$ consists of the edges incident to a vertex of $X$ and $M'_Y$ consists of the edges incident to a vertex of $Y$. Let $M''$ be the larger of $M'_X$ and $M'_Y$; then $|M''| \ge |M| / 2$. Note that $M''$ is a matching of $G''[A''_e,\overline{A''_e}]$ since $M'' \subseteq M'$.

By \cite[Lemma~9]{BCM15},  there exists some edge $e \in E(T)$ such that $G[A_e,\overline{A_e}]$ has a (not necessarily induced) matching $M$ of size at least $\tw(G)/3$. By the previous paragraph, $G''[A''_e,\overline{A''_e}]$ has a matching $M''$ of size at least $|M| / 2 \ge \tw(G)/6$, which consists of edges between a vertex in $Z$ and a vertex in either $X$ or $Y$.

We claim that $M''$ is an induced matching. Suppose not.  Then we may assume (up to swapping $X$ and $Y$) that $M''$ has edges $x_uz_{uv}$ and $x_{u'}z_{u'v'}$, for some distinct $u,u' \in V(G)$, and $G''$ also has an edge $x_uz_{u'v'}$ or $x_{u'}z_{uv}$. But, by construction, the vertices $z_{uv},z_{u'v'} \in Z$ have only one neighbour in $X$, so neither $x_uz_{u'v'}$ nor $x_{u'}z_{uv}$ is an edge of $G''$. Thus $M''$ is induced, and hence $\mimw_{G''}(T'',\delta'') \ge \tw(G)/6$, as required.
\end{proof}

We use Lemma~\ref{modified-bcmconstruction} to show the following theorem.

\begin{theorem}
    \label{diamond-p6}
    The class of $(K_4,\diamond,P_6,P_2+P_4)$-free graphs has unbounded mim-width.
\end{theorem}

\begin{proof}
We show that for every integer $k$, there is a $(K_4,\diamond,P_6,P_2+P_4)$-free graph~$G$ such that $\mimw(G) > k$.  Let $B$ be a (simple) bipartite graph with $\tw(B) > 6k$ and let $G=B''$. Then $\mimw(G) > k$ by \cref{modified-bcmconstruction}. Observe that $X$, $Y$ and $Z$ are independent sets.

First we claim that $G$ is $K_4$-free. Suppose $G[A] \cong K_4$ for some $A \subseteq V(G)$. Since each vertex in $Z$ has degree~$2$, $A \subseteq X \cup Y$. But then $|A \cap X| \ge 2$ or $|A \cap Y| \ge 2$, a contradiction.

Next we claim that $G$ is $\diamond$-free. Suppose $G[A] \cong \diamond$ for some $A \subseteq V(G)$.  Since each vertex in $Z$ has degree~$2$, the degree-$3$ vertices of the diamond must be in $X$ or $Y$.  Since these vertices are adjacent, one is in $X$ and one is in $Y$.  As the other two vertices of the diamond are complete to these two vertices, these vertices are in $Z$. Let $A \cap X = \{x_u\}$, $A \cap Y = \{y_v\}$, and $A \cap Z = \{z_e,z_{e'}\}$.  Now $x_uz_ey_v$ and $x_uz_{e'}y_v$ are paths in $G$, corresponding to multiple edges $e=uv$ and $e'=uv$ in $B$, but this contradicts that $B$ is simple.

Next we claim that $G$ is $P_2+P_4$-free. Suppose $G[A] \cong P_2+P_4$ for some $A \subseteq V(G)$ and $G[A'] \cong P_4$ for some $A' \subseteq A$.If $A' \subseteq Y \cup Z$, then one end of $G[A']$ is in $Y$, and the other end is in $Z$. But each vertex in $Z$ has one neighbour in $X$ and one neighbour in $Y$, so $A' \cap X \neq \varnothing$ and, by symmetry, $A' \cap Y \neq \varnothing$. Now each vertex in $X$ or $Y$ is adjacent to a vertex of $G[A']$. So $A \setminus A' \subseteq Z$, but then $G[A \setminus A'] \cong 2P_1$, a contradiction.

It remains to show that $G$ is $P_6$-free. Suppose $G[A] \cong P_6$ for some $A \subseteq V(G)$. If $A \subseteq X \cup Z$, then each vertex of $A \cap Z$ has degree at most~$1$ in $G[A]$, so there are at most two such vertices. But then $|A \cap X| \ge 4$, and this set is independent in $G[A]$, a contradiction. So $A \cap Y \neq \varnothing$ and, by symmetry, $A \cap X \neq \varnothing$. Since $X$ is complete to $Y$, we also have $|A \cap (X \cup Y)| \le 3$. Without loss of generality we may assume $A \cap X$ is a singleton~$\{x\}$. Then $x$ has two neighbours in $A \cap Y$, so $A \cap X$ and $A \cap Z$ are anticomplete.   But then $A\cap (X\cup Z)$ is an independent set of size at least~$4$, a contradiction.
\end{proof}

\section{State of the Art}\label{s-soa}

In this section, we show the consequences of the results from Sections~\ref{s-mim}--\ref{sec:unbounded} for the boundedness and unboundedness of mim-width of classes of $(H_1,H_2)$-free graphs. We will also make a comparison between the results for mim-width and clique-width. 
In contrast to the situation where only one induced subgraph is forbidden, we note many differences when two induced subgraphs $H_1$ and $H_2$ are forbidden. 
Figure~\ref{f-sec4} illustrates a number of graphs that we use throughout the section.

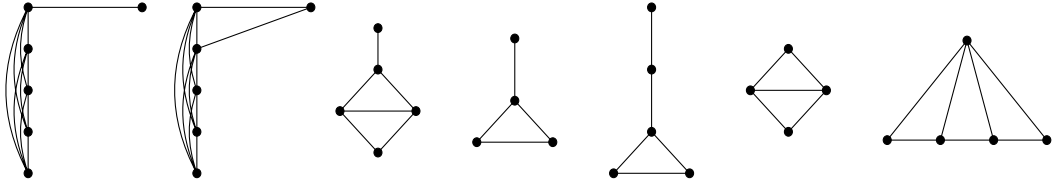
\begin{figure}[h]
  \centering
\begin{minipage}[c]{0.1\textwidth}
\begin{tikzpicture}[xscale=0.5, yscale=0.55]
\draw (-2,-2)--(-2,2)--(1,2)
(-2,-2) to[out=110,in=250] (-2,0) (-2,-1) to[out=110,in=250] (-2,1) (-2,0) to[out=110,in=250] (-2,2) 
(-2,-2) to[out=115,in=245] (-2,1) (-2,-1) to[out=115,in=245] (-2,2) (-2,-2) to[out=120,in=240] (-2,2);
\draw[fill=black] (-2,-2) circle [radius=3pt] (-2,-1) circle [radius=3pt] (-2,0) circle [radius=3pt] 
(-2,1) circle [radius=3pt] (-2,2) circle [radius=3pt] (1,2) circle [radius=3pt];
\end{tikzpicture}
\end{minipage}
\qquad
\begin{minipage}[c]{0.1\textwidth}
\begin{tikzpicture}[xscale=0.5, yscale=0.55]
\draw (-2,-2)--(-2,2)--(1,2)--(-2,1)
(-2,-2) to[out=110,in=250] (-2,0) (-2,-1) to[out=110,in=250] (-2,1) (-2,0) to[out=110,in=250] (-2,2) 
(-2,-2) to[out=115,in=245] (-2,1) (-2,-1) to[out=115,in=245] (-2,2) (-2,-2) to[out=120,in=240] (-2,2);
\draw[fill=black] (-2,-2) circle [radius=3pt] (-2,-1) circle [radius=3pt] (-2,0) circle [radius=3pt] 
(-2,1) circle [radius=3pt] (-2,2) circle [radius=3pt] (1,2) circle [radius=3pt];
\end{tikzpicture}
\end{minipage}
\qquad
\begin{minipage}[c]{0.07\textwidth}
\begin{tikzpicture}[xscale=0.5, yscale=0.55]
\draw (-1,0)--(0,1)--(1,0)--(-1,0)--(0,-1)--(1,0) (0,1)--(0,2);
\draw[fill=black] (-1,0) circle [radius=3pt] (1,0) circle [radius=3pt] 
(0,1) circle [radius=3pt] (0,-1) circle [radius=3pt] (0,2) circle [radius=3pt];
\end{tikzpicture}
\end{minipage}
\qquad
\begin{minipage}[c]{0.07\textwidth}
\begin{tikzpicture}[xscale=0.5, yscale=0.55]
\draw (0,1.5)--(0,0)--(-1,-1)--(1,-1)--(0,0);
\draw[fill=black] (-1,-1) circle [radius=3pt] (1,-1) circle [radius=3pt] 
(0,0) circle [radius=3pt] (0,1.5) circle [radius=3pt];
\end{tikzpicture}
\end{minipage}
\qquad
\begin{minipage}[c]{0.07\textwidth}
\begin{tikzpicture}[xscale=0.5, yscale=0.55]
\draw (0,3)--(0,1.5)--(0,0)--(-1,-1)--(1,-1)--(0,0);
\draw[fill=black] (-1,-1) circle [radius=3pt] (1,-1) circle [radius=3pt] (0,0) circle [radius=3pt] 
(0,1.5) circle [radius=3pt] (0,3) circle [radius=3pt];
\end{tikzpicture}
\end{minipage}
\qquad
\begin{minipage}[c]{0.07\textwidth}
\begin{tikzpicture}[xscale=0.5, yscale=0.55]
\draw (-1,0)--(0,1)--(1,0)--(-1,0)--(0,-1)--(1,0);
\draw[fill=black] (-1,0) circle [radius=3pt] (1,0) circle [radius=3pt] 
(0,1) circle [radius=3pt] (0,-1) circle [radius=3pt];
\end{tikzpicture}
\end{minipage}
\qquad
\begin{minipage}[c]{0.14\textwidth}
\begin{tikzpicture}[xscale=0.5, yscale=0.55]
\draw (0.7,-1)--(0,1.4)--(2.1,-1)--(2.1,-1)--(-2.1,-1)--(0,1.4)--(-0.7,-1);
\draw[fill=black] (-2.1,-1) circle [radius=3pt] (-0.7,-1) circle [radius=3pt] 
(0.7,-1) circle [radius=3pt] (2.1,-1) circle [radius=3pt] (0,1.4) circle [radius=3pt];
\end{tikzpicture}
\end{minipage}
\caption{The graphs $K_5\boxminus P_1=\overline{K_{1,4}+P_1}$, $\overline{K_{1,3}+2P_1}$, $\overline{S_{1,1,2}}$, $\paw$, $\hammer$, $\diamond$ and $\gem$.}\label{f-sec4}
\end{figure}

\subsection{Two Summary Theorems}

In our first summary theorem we give all pairs $(H_1,H_2)$ for which the mim-width of the class of $(H_1,H_2)$-free graphs is bounded.
This theorem gives more bounded cases than the corresponding summary theorem for boundedness of {\it clique-width} of classes of $(H_1,H_2)$-free graphs, which can be found in~\cite{DJP19} and which we need for our proof. To get the summary theorem for clique-width, replace Cases~(x)--(xv) of Theorem~\ref{t-sum1} by the more restricted case where  $H_1=K_s$ and $H_2=tP_1$ for some $s,t\geq 1$.

\begin{theorem}\label{t-sum1}
For graphs $H_1$ and $H_2$, the mim-width of the class of $(H_1,H_2)$-free graphs is \emph{bounded} and \emph{quickly computable} if one of the following holds:
\begin{enumerate}[(i)]
\item $H_1$ or $H_2 \ssi P_4$,
\item 
$H_1 \ssi \paw$ and $H_2 \ssi K_{1,3}+\nobreak 3P_1,\; K_{1,3}+\nobreak P_2,\;\allowbreak P_1+\nobreak P_2+\nobreak P_3,\;\allowbreak P_1+\nobreak P_5,\;\allowbreak P_1+\nobreak S_{1,1,2},\;\allowbreak P_2+\nobreak P_4,\;\allowbreak P_6,\; \allowbreak S_{1,1,3}$ or~$S_{1,2,2}$,
\item $H_1\ssi P_1+P_3$ and $H_2 \ssi \overline{K_{1,3}+\nobreak 3P_1},\; \overline{K_{1,3}+\nobreak P_2},\;\allowbreak \overline{P_1+\nobreak P_2+\nobreak P_3},\;\allowbreak \overline{P_1+\nobreak P_5},\;\allowbreak \overline{P_1+\nobreak S_{1,1,2}},\;\allowbreak \overline{P_2+\nobreak P_4},\;\nobreak \overline{P_6},\; \allowbreak \overline{S_{1,1,3}}$ or~$\overline{S_{1,2,2}}$,
\item
$H_1 \ssi \diamond$ and $H_2\ssi P_1+\nobreak 2P_2,\; 3P_1+\nobreak P_2$ or~$P_2+\nobreak P_3$,
\item
$H_1 \ssi 2P_1+P_2$ and $H_2\ssi \overline{P_1+\nobreak 2P_2},\; \overline{3P_1+\nobreak P_2}$ or~$\overline{P_2+\nobreak P_3}$,
\item 
$H_1 \ssi \gem$ and $H_2 \ssi P_1+\nobreak P_4$ or~$P_5$,
\item 
$H_1 \ssi P_1+P_4$ and $H_2 \ssi \overline{P_5}$,
\item $H_1\ssi K_3+\nobreak P_1$ and $H_2 \ssi K_{1,3}$,
\item $H_1\ssi 2P_1+\nobreak P_3$ and $H_2\ssi \overline{2P_1+\nobreak P_3}$,
\item $H_1\ssi 2P_1+P_2$ and $H_2\ssi \bowtie$,
\item $H_1\ssi K_{1,3}$ and $H_2\ssi 2P_2$,
\item $H_1\ssi K_r$ for $r\geq 1$ and $H_2 \ssi sP_1+P_5$ for $s\geq 0$, 
\item $H_1\ssi K_r \boxminus rP_1$ for $r\geq 1$ and $H_2\ssi 2P_2$,
\item $H_1\ssi  K_r \boxminus P_1$ for $r\geq 1$ and $H_2\ssi tP_2$ for  $t\geq 1$, or
\item $H_1\ssi K_r \boxminus K_r$ for $r\geq 1$ and $H_2\ssi sP_1+P_2$ for $s\geq 0$. 
\end{enumerate}
\end{theorem}

\begin{proof}
Cases (i)--(ix) follows from the fact that each of the classes of $(H_1,H_2)$-free graphs in these cases has bounded clique-width and that clique-width is quickly computable for general graphs~\cite{OS06}.
For Case~(i) we also refer to Theorem~\ref{t-p4}.
Boundedness of clique-width has been proven for Case (ii) as follows: in~\cite{DP16} for $K_{1,3}+\nobreak 3P_1$; in~\cite{DLRR12} for $K_{1,3}+\nobreak P_2$; in~\cite{DDP17} for $P_1+P_2+P_3$ and $P_1+P_5$;  in~\cite{DP16} for $P_1+\nobreak S_{1,1,2}$; in~\cite{DLP20} for $P_2+P_4$; in~\cite{BKM06} for $P_6$;  in~\cite{DLRR12} for  $S_{1,1,3}$; and in~\cite{DDP17} for $S_{1,2,2}$. It has been proven for Case~(iv) as follows: in~\cite{DDP17} for $P_1+2P_2$; and in~\cite{DHP19} for $3P_1+P_2$ and~$P_2+P_3$. It has been been proven for Case~(vi) as follows: in~\cite{BLM04} for $P_1+P_4$; and in~\cite{BLM05} for $P_5$. It has been proven for Case~(viii) and~(ix) in~\cite{BL02,BM02} and~\cite{BDJLPZ20}, respectively. Cases~(iii),~(v),~(vii) follow from Cases~(ii),~(iv) and~(vi), respectively, after recalling that the clique-width of a class of $(H_1,H_2)$-free graphs is bounded if and only if the  clique-width of the class of $\left(\overline{H_1},\overline{H_2}\right)$-free graphs is bounded~\cite{KLM09}. 
Cases~(x) and~(xi) follow from Theorems~\ref{t-3p1bowtie} and~\ref{t-2p2claw} respectively.
Case (xii) has been proven in~\cite{BHP}.
Cases (xiii)--(xv) follow from Theorems~\ref{thmbounded1}--\ref{thmbounded3}, respectively.
\end{proof}

\noindent
For our second summary theorem, we turn to the unbounded cases. 
We let ${\cal S}$ be the class of graphs every connected component of which is either a subdivided claw or a path.
We let  ${\cal N}$ denote the class of graphs that contain a connected component with 
either a cycle of length at least~$4$ or
at least two (not necessarily vertex-disjoint) triangles; note, for example, that  ${\cal N}$ contains $C_4$, $\diamond$, and $K_4$. 

\begin{theorem}\label{t-sum2}
For graphs $H_1$ and $H_2$, the class of $(H_1,H_2)$-free graphs has \emph{unbounded} mim-width if one of the following holds:
\begin{enumerate}[(i)]
\item $H_1\notin {\cal S}$ and $H_2\notin {\cal S}$,
\item $H_1\si C_3$ and $H_2 \si P_3+P_6$, $P_8$ or $S_{1,1,5}$,
\item $H_1\si K_{1,3}$ and $H_2\in {\cal N}$,
\item $H_1\si \diamond$ and $H_2\si 5P_1$, $P_2+P_4$, $2P_3$ or $P_6$,
\item $H_1\si 3P_1$ and $H_2\si 3P_1$, $C_5$ or $\overline{C_{2s+1}}$ for $s\geq 3$,
\item $H_1\si 4P_1$ and $H_2\si \gem$, $\overline{3P_1+P_2}$ or $\overline{P_1+2P_2}$,
\item $H_1\si 2P_2$ and $H_2\si C_4$, $C_5$, $K_{1,4}$, $2P_2$, $\overline{3P_1+P_2}$ or $\sun_t$ for $t\geq 3$, or
\item $H_1\si K_4$ and $H_2\si P_2+P_4$ or $P_6$.
\end{enumerate}
\end{theorem}

\begin{proof}
Cases~(i) and (iii) follow from Theorem~\ref{walls} and Lemma~\ref{netwalls}, respectively, possibly after applying Lemma~\ref{subdivision} a sufficient number of times. All three subcases of Case~(ii) follows from Theorem~\ref{cb-p8}. The first subcase of Case~(iv) follows from Theorem~\ref{diamond-5p1}, the second one follows from Theorem~\ref{diamond-p6}, the third one follows from Theorem~\ref{diamond-2p3} and the fourth one follows from Theorem~\ref{diamond-p6}. All three subcases of Case~(v) follow from Lemma~\ref{cobipartite}. Case~(vi) follows from \cref{gem-4p1,k5minus-4p1}. All subcases of Case~(vii) follow from Lemma~\ref{split}. 
Case~(viii) follows from Theorem~\ref{diamond-p6}.
\end{proof}

\noindent
We note that the situation for the unbounded cases is again different from the situation for the unbounded cases of clique-width. For example, $(H_1,H_2)$-free graphs have unbounded clique-width if both $\overline{H_1}\notin {\cal S}$ and $\overline{H_2}\notin {\cal S}$ (see, for example,~\cite{DP16}). Take, 
for instance, 
$H_1=4P_1$ and $H_2=2P_2$. Then $\overline{H_1}=K_4$ and $\overline{H_2}=C_4$, and thus
$\overline{H_1}\notin {\cal S}$ and $\overline{H_2}\notin {\cal S}$, so $(H_1,H_2)$-free graphs have unbounded clique-width. However, by Theorem~\ref{t-sum1}-(xiii), $(H_1,H_2)$-free graphs have bounded mim-width.
As $(\overline{H_1},\overline{H_2})$-free graphs have unbounded mim-width by Theorem~\ref{t-sum2}-(i),
this example also shows that the complementation operation, a standard tool for working with clique-width, does not preserve mim-width. Consequently, 
for mim-width there are many more open cases  than the only five open cases for clique-width~\cite{DJP19}.

\subsection{Three Consequences of the Summary Theorems}

In order to get a handle on the open cases for mim-width, we now present some consequences of Theorems~\ref{t-sum1} and~\ref{t-sum2}.  We first consider the case where $H_1$ and $H_2$ are forests.

\begin{corollary}\label{c-1}
Let $H_1$ and $H_2$ be \emph{forests}.
Either the pair $(H_1,H_2)$ satisfies \cref{t-sum1} or \cref{t-sum2}, or one of the following holds:
\begin{enumerate}
\item $H_1=2P_2$ and $H_2=K_{1,3} + sP_1$ for $s \geq 1$, or
\item $H_1=2P_2$ and $H_2=S_{1,1,2} + sP_1$ for $s\geq 0$.
\end{enumerate}
\end{corollary}

\begin{proof}
Throughout the proof we assume that $H_1$ and $H_2$ are not induced subgraphs of $P_4$, as otherwise we can apply  Theorem~\ref{t-sum1}-(i). This means that $H_1$ contains an induced $3P_1$ or an induced $2P_2$ and the same holds for $H_2$. If both contain an induced $3P_1$, then we can apply Theorem~\ref{t-sum2}-(v). If both contain an induced $2P_2$, then we can apply Theorem~\ref{t-sum2}-(vii). Suppose neither of these two cases apply. Then we may assume without loss of generality that $2P_2\ssi H_1$ while $3P_1\not\ssi H_1$, and $3P_1\ssi H_2$ while $2P_2\not\ssi H_2$. The above implies that $H_1=2P_2$ and $H_2$ has at most one connected component with an edge.

First suppose that $H_2$ is a linear forest. Then $H_2=sP_1+P_3$ or $H_2=sP_1+P_4$ for some $s\geq 1$, and we apply Theorem~\ref{t-sum1}-(xiii). Now suppose that $H_2$ is not a linear forest, so $K_{1,3}\ssi H_2$. If $K_{1,4}\ssi H_2$, then we apply Theorem~\ref{t-sum2}-(vii). 
If $H_2=K_{1,3}$, then we apply Theorem~\ref{t-sum1}-(xi).
Hence $H_2=K_{1,3} + sP_1$ for some $s\geq 1$ or $H_2=S_{1,1,2} + tP_1$ for some $t\geq 0$.
\end{proof}

\begin{open}\label{o-1}
Determine the (un)boundedness of mim-width of $(H_1,H_2)$-free graphs when
\begin{enumerate}
\item $H_1=2P_2$ and $H_2=K_{1,3} + sP_1$ for $s \geq 1$, or
\item $H_1=2P_2$ and $H_2=S_{1,1,2} + sP_1$ for $s\geq 0$.
\end{enumerate}
\end{open}

\noindent
Next we consider the case where $H_1$ and $H_2$ are connected.

\begin{corollary}\label{c-2}
Let $H_1$ and $H_2$ be \emph{connected} graphs.
Either the pair $(H_1,H_2)$ satisfies \cref{t-sum1} or \cref{t-sum2}, or one of the following holds:
\begin{enumerate}
\item $H_1=P_5$ and $H_2=\overline{S_{1,1,2}}$ or $\overline{K_{1,r}+sP_1}$ for $r\geq 3$ and $s \in \{1,2\}$,
\item $H_1=P_7$ or $S_{h,i,j}$ for $h\leq i\leq j\leq 4$ with $i+j\leq 6\leq h+i+j$ and $H_2=C_3$ or $\paw$, or
\item $H_1=K_{1,3}$ or $S_{1,1,2}$ and $H_2=\hammer$.
\end{enumerate}
\end{corollary}

\begin{proof}
If $H_1\notin {\cal S}$ and $H_2\notin {\cal S}$, then we apply Theorem~\ref{t-sum2}-(i). Hence, we may assume without loss of generality that $H_1\in {\cal S}$. As $H_1$ is connected, this means that $H_1$ is a subdivided claw or a path. If $H_1$ is $3P_1$-free, then $H_1\ssi P_4$, and we apply Theorem~\ref{t-sum1}-(i). Assume that $3P_1\ssi H_1$. Then $H_2$ must be co-bipartite, as otherwise we can apply Theorem~\ref{t-sum2}-(v). 

First suppose $H_1$ is a path. If $H_1\ssi P_4$, then we apply Theorem~\ref{t-sum1}-(i). Now suppose $P_5\ssi H_1$. Then both $3P_1\ssi H_1$ and $2P_2\ssi H_1$. Then $H_2$ must be a co-bipartite $\overline{3P_1+P_2}$-free split graph, as otherwise we can apply  Theorem~\ref{t-sum2}-(vii). Suppose $H_1=P_5$. If $H_2=\gem$, then we apply Theorem~\ref{t-sum1}-(vi). If $H_2=K_r$ for any $r\geq 1$,
 then we apply Theorem~\ref{t-sum1}-(xii). 
Otherwise we find that  $H_2=\overline{S_{1,1,2}}$ or  $H_2=\overline{K_{1,r}+sP_1}$ for some $r\geq 3$ and $s \in \{1,2\}$, which correspond to Case~1. Now suppose $H_1=P_6$. If $K_4\ssi H_2$, then we apply Theorem~\ref{t-sum2}-(viii). Suppose $H_2$ is $K_4$-free. If $H_2\ssi \paw$, then we apply Theorem~\ref{t-sum1}-(ii). Otherwise $\diamond \ssi H_2$ and we apply Theorem~\ref{t-sum2}-(iv). Now suppose $H_1=P_7$.  If $K_4\ssi H_2$ or $\diamond\ssi H_2$, then we apply Theorem~\ref{t-sum2}-(viii) or Theorem~\ref{t-sum2}-(iv), respectively. Otherwise we find that $H_2=C_3$ or $\paw$; this falls under Case~2. Finally suppose $P_8\ssi H_1$. If $C_3\ssi H_2$, then we apply Theorem~\ref{t-sum2}-(ii). Otherwise we find that $H_2\ssi P_4$ and we apply Theorem~\ref{t-sum1}-(i).

Now suppose $H_1$ is a subdivided claw. If $C_4$, $K_4$, or $\diamond \ssi H_2$, then we apply Theorem~\ref{t-sum2}-(iii). From now on assume that $H_2$ is $(C_4,K_4,\diamond)$-free. Recall that $H_2$ is co-bipartite. If $H_2$ is $C_3$-free, this means that $H_2\ssi P_4$ and we apply Theorem~\ref{t-sum1}-(i). Hence, we may assume that $C_3\ssi H_2$.  This means that $H_2\in \{C_3,\paw,\bowtie,\hammer,2C_3+e\}$, where the graph $2C_3+e$ is obtained from $2C_3$ by inserting an edge between the two triangles. First suppose $H_1\in \{K_{1,3},S_{1,1,2}\}$. If $H_2\ssi \paw$, then we apply Theorem~\ref{t-sum1}-(ii).  Otherwise we find that $H_2\in \{\bowtie,\hammer,2C_3+e\}$. If $H_2\in \{\bowtie,2C_3+e\}$, then we apply Theorem~\ref{t-sum2}-(iii). The two remaining cases correspond to Case~3. Now suppose that $H_1\notin \{K_{1,3},S_{1,1,2}\}$. Then $2P_2\ssi H_1$. If $H_2\in \{\bowtie,\hammer,2C_3+e\}$, then $2P_2\ssi H_2$, which means that we can apply Theorem~\ref{t-sum2}-(vii). Hence, we may assume that $H_2\in \{C_3,\paw\}$. If $H_1\in \{S_{1,2,2},S_{1,1,3}\}$, then we apply Theorem~\ref{t-sum1}-(ii). If $H_1$ is not $(P_3+P_6,P_8,S_{1,1,5})$-free, then we apply Theorem~\ref{t-sum2}-(ii). Otherwise we obtain the remaining cases of Case~2.
\end{proof}

\begin{open}\label{o-2}
Determine the (un)boundedness of mim-width of $(H_1,H_2)$-free graphs when
\begin{enumerate}
\item $H_1=P_5$ and $H_2=\overline{S_{1,1,2}}$ or $\overline{K_{1,r}+sP_1}$ for $r\geq 3$ and $s \in \{1,2\}$,
\item $H_1=P_7$ or $S_{h,i,j}$ for $h\leq i\leq j\leq 4$ with $i+j\leq 6\leq h+i+j$ and $H_2=C_3$ or $\paw$, or
\item $H_1=K_{1,3}$ or $S_{1,1,2}$ and $H_2=\hammer$.
\end{enumerate}
\end{open}

\noindent
Finally, we note that Theorems~\ref{t-sum1} and~\ref{t-sum2} cover all pairs $(H_1,H_2)$ with  $|V(H_1)| + |V(H_2)| \le 8$.

\begin{corollary}\label{c-8}
  Let $H_1$ and $H_2$ be graphs with $|V(H_1)| + |V(H_2)| \le 8$.
  Then the pair $(H_1,H_2)$ satisfies \cref{t-sum1} or \cref{t-sum2}.
\end{corollary}

\begin{proof}
If $H_1\notin {\cal S}$ and $H_2\notin {\cal S}$, then we apply Theorem~\ref{t-sum2}-(i). Hence, we may assume without loss of generality that $H_1\in {\cal S}$.
As each of the pairs $(H_1,H_2)$ in Open Problem~\ref{o-1} (Corollary~\ref{c-1})
has $|V(H_1)| + |V(H_2)| \ge 9$, we deduce that $H_2$ contains a cycle.
As each of the pairs $(H_1,H_2)$ in Open Problem~\ref{o-2} (Corollary~\ref{c-2}) has $|V(H_1)| + |V(H_2)| \ge 9$, we deduce that at least one of $H_1$, $H_2$ is disconnected.

\medskip
\noindent
{\bf Case 1.} $H_1$ is disconnected.\\
First suppose  that $H_1$ is $3P_1$-free. Then either $H_1\ssi P_4$ or $H_1=2P_2$. In the first case we apply Theorem~\ref{t-sum1}-(i),
Assume the latter case. Then $H_2$ is $C_4$-free, as otherwise we apply
Theorem~\ref{t-sum2}-(vii). Hence $H_2$ contains a $C_3$.
If $H_2\in \{C_3,K_3+P_1,K_3\boxminus P_1,K_4\}$, then we apply Theorem~\ref{t-sum1}-(xiii). Otherwise, $H_2=\diamond$ and we apply Theorem~\ref{t-sum1}-(iv).

Now suppose $H_1$ contains an induced $3P_1$.
Then $H_2$ must be $3P_1$-free, as otherwise we can apply Theorem~\ref{t-sum2}-(v).
First consider when $|V(H_1)| \le 4$ and $|V(H_2)| \le 4$.
Then $H_1\in \{3P_1,4P_1,2P_1+P_2,P_1+P_3\}$ and $H_2\in \{C_3,C_4,\diamond,\paw,K_3+P_1,K_4\}$.
If $H_1=P_1+P_3$, then we apply Theorem~\ref{t-sum1}-(iii).
So $H_1\in \{3P_1,4P_1,2P_1+P_2\}$.
If $H_2\in \{C_3,C_4,\paw,K_3+P_1,K_4\}$, then we apply Theorem~\ref{t-sum1}-(xiv) or Theorem~\ref{t-sum1}-(xv); whereas if $H_2=\diamond$, then we apply Theorem~\ref{t-sum1}-(iv).

It remains to consider when $H_1 = 3P_1$ and $|V(H_2)| = 5$, or $H_2 = C_3$ and $|V(H_1)| = 5$.
In the latter case, $H_2 = C_3$ and $H_1$ is a linear forest on 5 vertices, in which case we apply \cref{t-sum1}-(ii).
In the former case, if $H_2 \in \{K_3+P_2, \hammer, \overline{P_5}, K_4+ P_1, K_4 \boxminus  P_1, K_5\}$, then $H_2 \ssi K_5 \boxminus K_5$, so we apply \cref{t-sum1}-(xv); whereas if $H_2$ belongs to $\{\overline{S_{1,1,2}}, \overline{P_2+P_3}, \gem, \overline{P_1+2P_2}, \overline{2P_1+P_3}, \overline{3P_1+P_2}\}$, then we apply \cref{t-sum1}-(iii).
The only possibility that remains is $H_2=\bowtie$, for which we apply \cref{t-sum1}-(x).

\medskip
\noindent
{\bf Case 2.} $H_1$ is connected.\\
Then $H_2$ is disconnected.
As $H_2$ contains a cycle, $|V(H_2)| \ge 4$, so $|V(H_1)| \le 4$.
As $H_1$ is connected and belongs to ${\cal S}$, we find that $H_1\ssi P_4$ or $H_1=K_{1,3}$.
In the first case we apply Theorem~\ref{t-sum1}-(i).
In the second case, $|V(H_1)|=4$, so $|V(H_2)| = 4$.
As $H_2$ is disconnected and contains a cycle, $H_2=K_3+P_1$, so we apply Theorem~\ref{t-sum1}-(viii).
\end{proof}

\subsection{When ${\mathbf{H_1}}$ is Complete or Edgeless}

We first consider the (un)boundedness of mim-width for the class of $(K_r,H_2)$-free graphs for a positive integer $r$ and a graph $H_2$. Such classes are interesting for the following reason. For any $H_2$ such that mim-width is bounded and quickly computable for the class of $(K_r,H_2)$-free graphs, \textsc{$k$-Colouring} is polynomial-time solvable for all $k < r$; for example, see \cite{BHP} for the case where $H_2 \ssi sP_1+P_5$.
More generally, for problems having polynomial-time algorithms when mim-width is bounded and quickly computable, we obtain $n^{f(\omega(G))}$-time algorithms, for some function $f$, 
when restricted to $H_2$-free graphs; 
that is, {\XP} algorithms parameterized by $\omega(G)$ (the size of the largest clique in $G$).
Recently, Chudnovsky et al.~\cite{CKPRS20} showed that for $P_5$-free graphs, there exists an $n^{O(\omega(G))}$-time algorithm for \textsc{Max Partial $H$-Colouring}, a problem generalizing \textsc{Maximum Independent Set} and \textsc{Odd Cycle Transversal}, and which is polynomial-time solvable when mim-width is bounded and quickly computable.

For $r \ge 4$, \cref{t-sum1,t-sum2} imply that the mim-width of the class of $(K_r, H_2)$-free graphs is bounded and quickly computable when $H_2 \ssi sP_1 +P_5$ or $tP_2$, 
and unbounded when $H_2 \si K_{1,3}$, $P_2 + P_4$, or $P_6$, or $H_2
\notin \mathcal{S}$.  In the following theorem we prove that all remaining cases belong to one infinite family:
when $H_2 = tP_2+uP_3$ for $u \ge 1$ and $t+u \ge 2$.
Note that Theorem~\ref{t-kr-h2} just concerns the case that $r \ge 4$.  When $r = 3$, further open cases arise; for example, see Open Problem~\ref{o-2}.

\begin{theorem}\label{t-kr-h2}
  Let $H$ be a graph and let $r \geq 4$ be an integer.
  Then exactly one of the following holds:

  \begin{itemize}
    \item $H \ssi sP_1+P_5$ or $tP_2$, and the mim-width of the class of
      $(K_r, H)$-free graphs is bounded and quickly computable;
    \item $H \notin \mathcal{S}$, or $H \si K_{1,3}$, $P_2+P_4$,
      or $P_6$, and the mim-width of the class of $(K_r, H)$-free graphs
      is unbounded; or
    \item $H = tP_2 + uP_3$ where $u \geq 1$ and $t+u \geq 2$.
  \end{itemize}
\end{theorem}

\begin{proof}
  By \cref{t-sum2}-(i), if $H \notin \mathcal{S}$, then the mim-width of the class of $(K_r,H)$-free graphs is unbounded.
  So we may assume that $H$ is a forest of paths and subdivided claws.
  By \cref{t-sum2}-(iii), if $H$ contains a $K_{1,3}$, then the mim-width is again unbounded.
  So we may assume that $H$ is a linear forest.
  If $H \ssi sP_1 + P_5$ or $H \ssi tP_2$, then mim-width is bounded
  and quickly computable by parts~(xii) and~(xiv) of \cref{t-sum1}.
  So we may assume that $H$ is a linear forest containing $P_2 + P_3$.
  By \cref{t-sum2}-(viii), we may also assume $H$ contains neither $P_2 + P_4$ nor $P_6$, otherwise the mim-width is again unbounded.
  It now follows that $H \ssi tP_2 + uP_3$ for some $u,t$ such that
  $u \geq 1$ and $t+u \geq 2$.
\end{proof}

\begin{open}\label{o-3}
  For an integer $r \ge 4$, and for each integer $t \geq 0$ and $u \geq 1$ such that $t+u \geq 2$, determine the (un)boundedness of the class of $(K_r, tP_2 + uP_3)$-free graphs.
\end{open}

\noindent
We note that this is also open when $r=3$, except when $u=t=1$ (so $H_2 = P_2+P_3$) in which case we can apply \cref{t-sum1}-(ii).

\medskip
\noindent
We now consider the class of $(rP_1,H_2)$-free graphs, for an integer $r$ and a graph~$H_2$. If the mim-width of such a class of graphs is bounded and quickly computable, we obtain, for many problems, {\XP} algorithms parameterized by $\alpha(G)$ for the class of $H_2$-free graphs, where $\alpha(G)$ is the size of the largest independent set in $G$.
For $r \ge 5$, \cref{t-sum1,t-sum2} imply that the mim-width of the class of $(rP_1, H_2)$-free graphs is bounded and quickly computable when $H_2 \ssi K_t \boxminus K_t$ for some $t$, and unbounded when $H_2$ is not co-bipartite, or $H_2 \si \diamond$.
Below we show that all unresolved cases belong to the infinite family $H_2 =
\overline{K_{s,t}+P_1}$ for $s,t \ge 2$
(we observe that if $s=t=2$, then $H_2 = \bowtie$).
Note that Theorem~\ref{t-rp1-h2} just concerns the case that $r \ge 5$.  When $r \in \{3,4\}$, further open cases arise, and there are more cases where the class of $(rP_1, H)$-free graphs has bounded mim-width, by cases (iii) and (x) of \cref{t-sum1}.

\begin{theorem}\label{t-rp1-h2}
Let $H$ be a graph and let $r \geq 5$ be an integer.
Then exactly one of the following holds:

\begin{itemize}
  \item $H \ssi K_t \boxminus K_t$ for some integer $t \ge 1$, and the
    mim-width of the class of $(rP_1,H)$-free graphs is bounded and
    quickly computable;
  \item $H$ is not co-bipartite or $H \si \diamond$, and the mim-width
    of the class of $(rP_1,H)$-free graphs is unbounded; or
  \item $H = \overline{K_{s,t} + P_1}$ for some $s,t \geq 2$.
\end{itemize}
\end{theorem}

\begin{proof}
  By \cref{t-sum2}-(v), if $H$ is not co-bipartite, then the mim-width of the class of $(rP_1,H)$-free graphs is unbounded.
  So we may assume that $H$ is co-bipartite.
  In particular, $H$ is $3P_1$-free, and hence if $H$ is a forest,
  we have that $H \ssi P_4$ or $H \ssi 2P_2$.
  In either case, $H \subseteq K_4 \boxminus K_4$, so the mim-width is bounded and quickly computable by \cref{t-sum1}-(i).
  So we may assume that $H$ contains a cycle.
  In particular, since $H$ is $(C_5,3P_1)$-free, $H$ contains no
  induced cycle of length at least $5$.
  By \cref{t-sum2}-(iv) we may assume that $H$ contains no
  diamond, otherwise the class has unbounded mim-width.
  
  Suppose that $H$ contains an induced $C_4$.
  It follows from $H$ being co-bipartite and diamond-free that
  $H \ssi K_t \boxminus K_t$ for some $t$, in which case mim-width is bounded and
  quickly computable by \cref{t-sum1}-(xv).
  So we may assume that $H$ does not contain an induced $C_4$, and hence
  $H$ is chordal.

  It remains to show that $H$ is a block graph consisting of two
  blocks each being complete and having at least 3 vertices.
  Let $K$ be a maximum clique of $H$.
  So $K$ has size at least 3.
  By \cref{t-sum1}-(xv) we may assume that $V(H) \setminus K \neq \varnothing$.
  Since $H$ is diamond-free and by the maximality of $K$, any vertex of $H$ not in $K$ has at
  most one neighbour in $K$.
  Then since $H$ is $3P_1$-free, $V(H) \setminus K$ is a clique.
  Now, if at most one vertex of $V(H) \setminus K$ has a neighbour
  in $K$, then $H$ is an induced subgraph of $K_r \boxminus K_r$,
  so we can apply \cref{t-sum1}-(xv).
  So we may assume there are distinct vertices
  $u, v \in V(H) \setminus K$ each with a single neighbour in $K$.
  Suppose that $N(u) \cap K = \{k_u\}$ and $N(v) \cap K = \{k_v\}$
  for distinct $k_u, k_v \in K$.
  Since $H$ is $3P_1$-free, $uv \in E(H)$.
  But then $\{u,v,k_u,k_v\}$ induces a $C_4$ in $H$, a contradiction.
  Without loss of generality, $N(V(H) \setminus K) \cap K \subseteq
  \{ k_u \}$.
  Now, since $H$ is diamond-free and $V(H) \setminus K$ is a clique, $V(H) \setminus K$ is complete to
  $\{k_u\}$.
  It follows that $H = \overline{K_{s,t}+P_1}$ for some $s,t \geq
  2$.
\end{proof}

\begin{open}\label{o-4}
  For each integer $r \ge 4$, and for each integer $s,t \geq 2$, determine the (un)boundedness of the class of $(rP_1, \overline{K_{s,t} + P_1})$-free graphs.
\end{open}

\noindent
We note that \cref{o-4} includes the case $r=4$, in contrast to \cref{t-rp1-h2}, since the (un)boundedness of $(4P_1,\overline{K_{s,t} + P_1})$-free graphs is also open for $s\geq 2$ and $t\geq 2$.
In fact, when $r=3$, the (un)boundedness of $(3P_1,\overline{K_{s,t} + P_1})$-free graphs is also open except when $s=t=2$, in which case we have the class of $(3P_1,\bowtie)$-free graphs, and so we can apply \cref{t-3p1bowtie}.

\section{Conclusion}
\label{sec:conclusion}

We extended the toolkit for proving (un)boundedness of mim-width of hereditary graph classes. 
Using the extended toolkit,
we found new classes of $(H_1,H_2)$-free graphs of bounded width and unbounded mim-width.
We showed that the situation for mim-width of hereditary graph classes is different from the situation for clique-width, even 
when only two induced subgraphs $H_1$ and $H_2$ are forbidden.
 For future work, Open Problems~\ref{o-1}--\ref{o-4} deserve attention. 
 In particular, the class of $(P_5,\overline{K_{1,r}+sP_1})$-free graphs, for $r\geq 3$ and $s \in \{1,2\}$ 
 (Case~1 of Open Problem~\ref{o-2}),  is the only remaining infinite family of pairs $(H_1,H_2)$ where both $H_1$ and $H_2$ are connected.
Moreover, for Open Problem~\ref{o-1}, a similar approach to \cref{t-2p2claw} might be conducive to resolving further open cases where $H_1=2P_2$.

\end{document}